\documentclass[a4paper,UKenglish,cleveref, autoref, thm-restate, pdfa]{lipics-v2021}

\pdfoutput=1 
\hideLIPIcs  

\usepackage{algorithm}
\usepackage[noend]{algpseudocode}

\newcommand{\eps}{\varepsilon}

\renewcommand{\leq}{\leqslant}
\renewcommand{\geq}{\geqslant}
\newtheorem{fact}[theorem]{\bfseries{Fact}}

\title{Weakly Approximating Knapsack in Subquadratic Time} 

\author{Lin Chen}{Zhejiang University, Hangzhou, China
\and \url{https://person.zju.edu.cn/en/linchen}}{chenlin198662@zju.edu.cn}{https://orcid.org/0000-0003-3909-4916}{}

\author{Jiayi Lian}{Zhejiang University, Hangzhou, China
\and \url{https://joeylian.github.io/}}{jiayilian@zju.edu.cn}{https://orcid.org/0009-0006-4408-9627}{}

\author{Yuchen Mao}{Zhejiang University, Hangzhou, China
\and \url{https://person.zju.edu.cn/en/maoyc}}{maoyc@zju.edu.cn}{https://orcid.org/0000-0002-1075-344X}{Supported by National Natural Science Foundation of China [Project No. 62402436].}

\author{Guochuan Zhang}{Zhejiang University, Hangzhou, China
\and \url{https://person.zju.edu.cn/en/0096209}}{zgc@zju.edu.cn}{https://orcid.org/0000-0003-1947-7872}{Supported by National Natural Science Foundation of China [Project No. 12271477].}

\authorrunning{L. Chen, J. Lian, Y. Mao and G. Zhang} 

\Copyright{Lin Chen, Jiayi Lian, Yuchen Mao and Guochuan Zhang} 



\relatedversion{} 

\nolinenumbers 

\EventEditors{Keren Censor-Hillel, Fabrizio Grandoni, Joel Ouaknine, and Gabriele Puppis}
\EventNoEds{4}
\EventLongTitle{52nd International Colloquium on Automata, Languages, and Programming (ICALP 2025)}
\EventShortTitle{ICALP 2025}
\EventAcronym{ICALP}
\EventYear{2025}
\EventDate{July 8--11, 2025}
\EventLocation{Aarhus, Denmark}
\EventLogo{}
\SeriesVolume{334}
\ArticleNo{117}

\begin{document}

\maketitle

\begin{abstract}
We consider the classic Knapsack problem.  Let $t$ and $\mathrm{OPT}$ be the capacity and the optimal value, respectively. If one seeks a solution with total profit at least $\mathrm{OPT}/(1 + \eps)$ and total weight at most $t$, then Knapsack can be solved in $\tilde{O}(n + (\frac{1}{\eps})^2)$ time [Chen, Lian, Mao, and Zhang '24][Mao '24]. This running time is the best possible (up to a logarithmic factor), assuming that $(\min,+)$-convolution cannot be solved in truly subquadratic time [Künnemann, Paturi, and Schneider '17][Cygan, Mucha, Węgrzycki, and Włodarczyk '19]. The same upper and lower bounds hold if one seeks a solution with total profit at least $\mathrm{OPT}$ and total weight at most $(1 + \eps)t$. Therefore, it is natural to ask the following question. 
    \begin{quote}
        If one seeks a solution with total profit at least $\mathrm{OPT}/(1+\eps)$ and total weight at most $(1 + \eps)t$, can Knsapck be solved in $\tilde{O}(n + (\frac{1}{\eps})^{2-\delta})$ time for some constant $\delta > 0$?
    \end{quote}
    We answer this open question affirmatively by proposing an $\tilde{O}(n + (\frac{1}{\eps})^{7/4})$-time algorithm.
\end{abstract}

\section{Introduction}
In the Knapsack problem, one is given a knapsack of capacity $t$ and a set of $n$ items. Each item $i$ has a weight $w_i$ and a profit $p_i$. The goal is to select a subset of items such that their total weight does not exceed $t$, and their total profit is maximized. 

Knapsack is one of Karp's 21 NP-complete problems~\cite{Kar72}, which motivates the study of approximation algorithms. In particular, it is well known that Knapsack admits fully polynomial-time approximation schemes (FPTASes), which can return a feasible solution with a total profit of at least $\mathrm{OPT}/(1 + \eps)$ in $\mathrm{poly}(n, \frac{1}{\eps})$ time. (We use $\mathrm{OPT}$ to denote the maximum total profit one can achieve using capacity $t$.) There has been a long line of research on designing faster FPTASes for Knapsack~\cite{IK75,Law79,KP04,Rhe15,Chan18,Jin19,DJM23,CLMZ24bSTOCKnapsack,
Mao24}. Very recently, Chen, Lian, Mao and Zhang~\cite{CLMZ24bSTOCKnapsack} and Mao~\cite{Mao24} indenpendently obtained $\tilde{O}(n + (\frac{1}{\eps})^2)$-time FPTASes\footnote{Thourghout, we use an $\tilde{O}(\cdot)$ notation to hide polylogarithmic factors in $n$ and $\frac{1}{\eps}$.}.
On the negative side, there is no $O((n + \frac{1}{\eps})^{2-\delta})$-time FPTAS for any constant $\delta > 0$ assuming $(\min, +)$-convolution cannot be solved in subquadratic time~\cite{CMWW19, KPS17}. Hence, further improvement on FPTASs for Knapsack seems hopeless unless the $(\min, +)$-convolution conjecture is overturned. Due to the symmetry of weight and profit, the same upper and lower bounds also hold if we seek a solution of total profit at least $\mathrm{OPT}$, but relax the capacity to $(1 + \eps)t$.
In this paper, we consider {\it weak approximation schemes}, which relax both the total profit and the capacity. More precisely, a weak approximation scheme seeks a subset of items with total weight at most $(1+\eps)t$ and total profit at least $\mathrm{OPT}/(1+ \eps)$.

Weak approximation algorithms are commonly studied in the context of resource augmentation (see e.g.,~\cite{FGJS05, IZ10}). One often resorts to resource augmentation when there is a barrier to further improvement on standard (i.e., non-weak) approximation algorithms. The two most relevant examples are Subset Sum and Unbounded Knapsack. Both problems admit an $\tilde{O}(n + (\frac{1}{\eps})^2)$-time FPTAS~\cite{KMPS03,BN21b,JK18}, while an $O((n+\frac{1}{\eps})^{2-\delta})$-time FPTAS has been ruled out assuming the $(\min,+)$-convolution conjecture~\cite{KPS17,CMWW19,MWW19,BN21b}. Their weak approximation, however, can be done in truly subquadratic time. For Subset Sum, Mucha, W{\k{e}}grzycki and W{\l}odarczyk~\cite{MWW19} showed a $O(n+(\frac{1}{\eps})^{5/3})$-time weak approximation scheme, and very recently, Chen et al.~\cite{CLMZ24cSTOCPartition} further improved the running time to $\tilde{O}(n+\frac{1}{\eps})$-time. 
For Unbounded Knapsack, Bringmann and Cassis~\cite{BC22} established an $\tilde{O}(n+(\frac{1}{\eps})^{3/2})$-time weak approximation scheme. Given that the standard approximation of Knapsack has been in a similar situation, it is natural to ask whether the weak approximation of Knapsack can be done in truly subquadratic time.

It is worth mentioning that weak approximation schemes for Subset Sum and Unbounded Knapsack are not straightforward extensions of the existing standard (non-weak) approximation schemes. The first truly subquadratic-time weak approximation scheme for Subset Sum~\cite{MWW19} as well as its recent improvement~\cite{CLMZ24cSTOCPartition} both rely on a novel application of additive combinatorics, and the first truly subquadratic-time weak approximation scheme for Unbounded Knapsack~\cite{BC22} relies on a breakthrough on bounded monotone $(\min, +)$-convolution~\cite{CL15,CDXZ22}. Unfortunately, techniques developed in these prior papers are not sufficient to obtain a truly subquadratic-time weak approximation scheme for Knapsack. 

\paragraph*{Our contribution.}
We propose the first truly subquadratic-time weak approximation scheme for Knapsack. 

\begin{theorem}
	Knapsack admits an $\tilde{O}({n+(\frac{1}{\eps})^{7/4}})$-time weak approximation scheme, which is randomized and succeeds with high probability.\footnote{Throughout, ``with high probability'' stands for ``with probability at least $1 - (n + \frac{1}{\eps})^{-\Omega(1)}$''.}
\end{theorem}

Our algorithm can be easily generalized to Bounded Knapsack, in which, every item $i$ can be selected up to $u_i$ times. When $u_i=O(1)$ for all $i$, Bounded Knapsack is equivalent to Knapsack (up to a constant factor), and therefore can be handled by our algorithm. For the general case, there is a standard approach of reducing Bounded Knapsack to instances where $u_i=O(1)$  (see e.g., ~\cite[Lemma 4.1]{MWW19}\cite[Lemma 2.2]{KX19}). Note that such a reduction will blow up item weights and profits by a factor $O(u_i)$, but this does not make much difference in (weak) approximation schemes due to scaling.

\subsection{Technical Overview}
We give a brief overview of the techniques that are used by our algorithm.

\paragraph*{Bounded Monotone (min,+)-Convolution} 
$(\min,+)$-convolution (more precisely, its counterpart, $(\max,+)$-convolution) has been used in almost all the recent progress on Knapsack~\cite{Chan18,Jin19,CLMZ24aSODA,Jin24,Bri24,CLMZ24bSTOCKnapsack,Mao24}. It allows one to cut the input instance into sub-instances, deal with each sub-instance independently, and then merge the results using $(\min,+)$-convolution. In general, it takes $O(n^2)$ time to compute the $(\min,+)$-convolution of two length-$n$ sequences\footnote{A slightly faster $O(n^2/2^{\Omega(\sqrt{\log n})})$-time algorithm is known by Bremner et.al.'s reduction to $(\min,+)$-matrix multiplication~\cite{BCD+14} and William's algorithm for the latter problem~\cite{Wil14} (which was derandomized by Chan and Williams~\cite{CW16}).}, and it is conjectured that there is no algorithm with time $O(n^{2 - \delta})$ for any constant $\delta > 0$. Nevertheless, some special cases of $(\min,+)$-convolution can be computed in subquadratic time (See~\cite{AKM+87,AT19,PRW21,CL15,CDXZ22} for example), and among them is bounded monotone $(\min,+)$-convolution where the two sequences are monotone and their entries are integers bounded by $O(n)$. In a celebrated work, Chan and Lewenstein~\cite{CL15} showed that bounded monotone $(\min,+)$-convolution can be computed in $O(n^{1.859})$ time. This was later improved to $\tilde{O}(n^{1.5})$ by Chi, Duan, Xie, and Zhang~\cite{CDXZ22}, and further generalized to rectangular monotone $(\min,+)$-convolution by Bringmann, D{\"{u}}rr and Polak~\cite{DBLP:conf/esa/BringmannD024}.

Bounded monotone $(\min,+)$-convolution is closely related to Unbounded Knapsack. Indeed, Bringmann and Cassis presented a $\tilde{O}(n+(\frac{1}{\eps})^{3/2})$-time weak approximation scheme~\cite{BC22} that builds upon bounded monotone $(\min,+)$-convolution, and they also showed an equivalence result between the two problems in the sense that a truly subquadratic-time weak approximation scheme for Unbounded Knapsack implies a truly subquadratic algorithm for bounded monotone $(\min,+)$-convolution, and vice versa.


For our problem, we follow the general framework mentioned above: we cut the input instance into sub-instances (called the reduced problem), deal with each sub-instance independently, and then merge the results using bounded monotone $(\min,+)$-convolution. Note that following this framework, for each sub-instance we need to compute the near-optimal objective value for all (approximate) knapsack capacities.

One might anticipate that a subquadratic algorithm for bounded monotone $(\min,+)$-convolution leads to a subquadratic weak approximation scheme for Knapsack straightforwardly, however, this is not the case. Bounded monotone $(\min,+)$-convolution is only used to merge results; the main challenge lies in dealing with each group, where we need the following two techniques.



\paragraph*{Knapsack with Items of Similar Efficiencies}
When all items have the same efficiency, Knapsack becomes Subset Sum, and it can be weakly approximated in $\tilde{O}(n + \frac{1}{\eps})$ time~\cite{CLMZ24cSTOCPartition}.  We extend the result to the case where all the items have their efficiencies (that is, profit-to-weight ratio) in the range $[\rho, \rho + \Delta]$ for constant $\rho$, and show that in such a case, Knapsack can be weakly approximated in $\tilde{O}(n + \frac{1}{\eps} + (\frac{1}{\eps})^2\Delta)$ time. This result may be of independent interest. The following is a rough idea. By scaling, we can assume that there are only $O(\frac{\Delta}{\eps})$ distinct item efficiencies. Then we divide the items into $O(\frac{\Delta}{\eps})$ groups so that the items in the same group have the same efficiency, and therefore, each group can be treated as a Subset Sum problem. For each group, we solve it by invoking Chen et al.'s $\tilde{O}(n + \frac{1}{\eps})$-time weak approximation scheme for Subset Sum~\cite{CLMZ24cSTOCPartition}. After solving all groups, we merge their results using 2D-FFT and divide-and-conquer, which takes $\tilde{O}((\frac{1}{\eps})^2 \Delta)$ time. 

We note that the general idea of using Subset Sum to solve Knapsack when items have similar efficiencies was used by Jin~\cite{Jin19}, and the idea of merging Subset Sum results using 2D-FFT and divide-and-conquer was used by Deng, Jin and Mao~\cite{DJM23}.

\paragraph*{Proximity Bounds}
Proximity is another tool that is crucial to the recent progress of Knapsack~\cite{PRW21,CLMZ24aSODA,Jin24,Bri24,CLMZ24bSTOCKnapsack}. Proximity bounds capture the intuition that the optimal solution should contain a lot of high-efficiency items and only a few low-efficiency items. (The efficiency of an item is defined to be its profit-to-weight ratio.) It was pioneered by Eisenbrand and Weismantel~\cite{EW19}. Their proximity bound works for a more general problem of integer programming and is based on Steinitz's lemma. The proof can be greatly simplified if one only considers Knapsack~\cite{PRW21}.  This bound was subsequently improved using tools from additive combinatorics~\cite{CLMZ24aSODA,Bri24,Jin24}. These bounds depend on the maximum item weight $w_{\max}$ or the maximum item profits $p_{\max}$, and leads to $\tilde{O}(n + w^2_{\max})$-time and $\tilde{O}(n + p^2_{\max})$-time algorithm for Knapsack~\cite{Jin24,Bri24}. (We remark that all the proximity bounds of this type assume that the item weights and profits are integers.) These results, however, seem not helpful in breaking the quadratic barrier: even if we can scale the $w_{\max}$ and $p_{\max}$ to integers bounded by $O(\frac{1}{\eps})$, it still takes  $\tilde{O}(n + (\frac{1}{\eps})^2)$ time to solve the problem.

We shall use another type of proximity bounds, which depends on the item efficiencies. Such a proximity bound was used to design FPTASes for Knapsack~\cite{Jin19,DJM23} before. This proximity states that the optimal solution selects almost all the high-efficiency items, and exchanges only a few of them for the low-efficiency items. Moreover, the number of exchanged items is inversely proportional to the efficiency gap between high- and low-efficiency items. Fewer exchanged items allow a larger additive error per exchanged item, and therefore, a faster approximation.

\paragraph*{Our Main Technical Contribution}
While each of the above-mentioned techniques appeared in previous papers before, it seems that none of them alone can break the quadratic barrier for weakly approximating Knapsack. Our main contribution is to modify and then combine these techniques to obtain a truly subquadratic-time algorithm. 

Note that if the item efficiencies differ significantly, then the efficiency-based proximity bound becomes sharp; If item efficiencies stay similar, then the Subset Sum based approach becomes efficient. To combine both ideas, the main difficulty is that we need to approximate the optimal objective value for all (approximate) capacities, and the efficiency-based proximity bound may be sharp for certain capacities, and weak for other capacities. Towards this, the main idea is to use different rounding precision for different items. We roughly explain it below. For simplicity, when we say we approximately solve a group of items, we mean to approximately compute the optimal objective value for all knapsack capacities when taking items from this group. We partition all items into multiple groups of fixed size. If the items in a group have similar efficiencies, then we approximately solve this group with high precision using the Subset Sum based approach; Otherwise, we approximately solve this group with low precision. Low precision means a larger error per item, but if proximity guarantees that only a few items from this group will contribute to the optimal solution (or only a few items will not contribute), then the total error may still be bounded. So the challenge is to prove that a sharp proximity holds for such a group regardless of the knapsack capacity.

We remark that a simpler implementation of the above idea leads to an algorithm with running time $\tilde{O}(n + (\frac{1}{\eps})^{11/6})$. By grouping the items in a more sophisticated way, the running time can be improved to  $\tilde{O}(n + (\frac{1}{\eps})^{7/4})$.


\subsection{Further Related Work}
\paragraph*{FPTASes with Running Time in Multiplicative Form} So far, we have discussed FPTASes for Knapsack whose running times are $n+f(\frac{1}{\eps})$. There is also a line of research focusing on FPTASes with a running time of the form $n^{O(1)}\cdot f(\frac{1}{\eps})$ (see, e.g.,~\cite{Law79,KP99,Chan18}). The best known running time of this form is $\tilde{O}(\frac{n^{3/2}}{\eps})$ due to Chan~\cite{Chan18}. It remains open whether there exists an FPTAS of running time $O(\frac{n}{\eps})$.  It is not even clear whether Knapsack admits a weak approximation scheme in $O(\frac{n}{\eps})$ time.

\paragraph*{Pseudo-polynomial Time Algorithm}
Recent years have witnessed important progress in the study of pseudo-polynomial time algorithms for Knapsack and Unbounded Knapsack~\cite{BC22,BC23,CLMZ24aSODA,Bri24,Jin24}. In particular, the best-known algorithm for Knapsack that is parametrized by the largest weight $w_{\max}$ (or the largest profit $p_{\max}$) runs in $\tilde{O}(n+w_{\max}^2)$-time (or $\tilde{O}(n+p_{\max}^2)$-time)~\cite{Bri24,Jin24}, and is the best possible assuming the $(\min, +)$-convolution conjecture~\cite{CMWW19, KPS17}. It is, however, a major open problem whether the quadratic barrier can be broken by combining both parameters, that is, whether an $\tilde{O}(n+{(w_{\max}+p_{\max})}^{2-\delta})$-time algorithm exists for Knapsack. Notably, for Unbounded Knapsack, one can find an algorithm that runs in $\tilde{O}(n+{(w_{\max}+p_{\max})}^{3/2})$ time~\cite{BC22}.

\subsection{Paper Outline}
In Section~\ref{sec:pre}, we introduce all the necessary notation and terminology. In Section~\ref{sec:reduce}, we reduce the problem so that it suffices to consider instances with nice properties. In Section~\ref{sec:tools}, we present all the ingredients that are needed by our algorithm. In Section~\ref{sec:weak}, we present an $\tilde{O}(n + (\frac{1}{\eps})^{11/6})$-time algorithm to illusrate our main idea. In Section~\ref{sec:strong}, we show how to improve the running time to $\tilde{O}(n + (\frac{1}{\eps})^{7/4})$. All the missing proofs can be found in the appendix.

\section{Preliminaries}\label{sec:pre}
We assume that $\eps \leq 1$, since otherwise we can approximate using a linear-time 2-approximation algorithm for Knapsack~\cite[Section 2.5]{KPD04}. We denote the set of real numbers between $a$ and $b$ as $[a, b]$. All logarithms are based $2$.

To ease the presentation of our algorithm, we interpret Knapsack from the perspective of tuples. An item can be regarded as a tuple $(w,p)$, where $w$ and $p$ represent the weight and profit, respectively. Throughout, we do not distinguish an item and a tuple, and always refer to the first entry as weight and the second entry as profit.  The efficiency of a tuple $(w, p)$ is defined to be its profit-to-weight ratio $\frac{p}{w}$.  The efficiency of $(0,0)$ is defined to be $0$.

Let $I$ be a multiset of tuples. We use $w(I)$ to denote the total weight of tuples in $I$. That is, $w(I) = \sum_{(w,p)\in I}w$. Similarly, we define $p(I) = \sum_{(w,p)\in I}p$. Define the set of all (2-dimensional) subset sums of $I$ as follows:
\[
    \mathcal{S}(I) = \{(w(I'), p(I')) : I' \subseteq I\}.
\]
Every tuple $(w, p) \in \mathcal{S}(I)$ corresponds to a Knapsack solution with total weight $w$ and total profit $p$. We say a tuple $(w, p)$ dominates another tuple $(w', p')$ if $w \leq w'$,  $p \geq p'$, and at least one of the two inequalities holds strictly. Let ${\cal S}^+(I)$ be the set obtained from ${\cal S}(I)$ by removing all dominated tuples.
Every tuple in ${\cal S}^+(I)$ corresponds to a Pareto optimal solution for Knapsack, and vice versa. Let $t$ be the given capacity. Knapsack can be formulated as follows:
\[
    \max\{\text{$p$ : $(w, p) \in \mathcal{S}^+(I)$ and $w \leq t$}\}.
\]

We consider the more general problem of (weakly) approximating $\mathcal{S}^+(I)$.
\begin{definition}\label{def:approx-set}
    Let $S$ and $S'$ be two sets of tuples. 
    \begin{romanenumerate}
        \item $S'$ approximates $S$ with factor $1 + \eps$ if for any $(w,p) \in S$, there exists $(w', p') \in S'$ such that $w' \leq (1 + \eps)w$ and $p' \geq p/(1 + \eps)$;

        \item  $S'$ approximates $S$ with additive error $(\delta_w, \delta_p)$ if for any $(w,p) \in S$, there exists $(w' , p') \in S'$ such that $w' \leq w + \delta_w$ and $p' \geq p - \delta_p$;

        \item  $S$ dominates $S'$ if for any $(w', p') \in S'$, there exists $(w, p) \in S$ such that $w \leq w'$ and $p \geq p'$.
    \end{romanenumerate}
\end{definition}
When $S'$ approximates a set $S$ that contains only one tuple $(w,p)$, we also say that $S'$ approximates the tuple $(w,p)$.

Let $(I,t)$ be a Knapsack instance. Let $\mathrm{OPT}$ be the optimal value. Our goal is to compute a set $S'$ such that 
\begin{itemize}
    \item $S'$ approximates $\mathcal{S}^+(I) \cap ([0, t]\times [0, OPT])$ with additive error $(\eps t, \eps \cdot \mathrm{OPT})$ and 

    \item $S'$ is dominated by $\mathcal{S}(I)$.
\end{itemize}
The former ensures that $S'$ provides a good (weak) approximation of the optimal solution, and the latter ensures that every tuple in $S'$ is not better than a true solution. Our definition is essentially the same as that by Bringmann and Cassis~\cite{BC22}, albeit their definition is for sequences.

Let $I_1 \cup I_2$ be a partition of $I$. It is easy to observe the following.
\begin{itemize}
\item \(
    {\cal S}(I) = {\cal S}(I_1) + {\cal S}(I_2)
\), where the sumset $X + Y$ is defined by the following formula.
\[
    X + Y = \{\text{$(w_1 + w_2, p_1 + p_2)$ : $(w_1, p_1) \in X$ and $(w_2, p_2) \in Y$}\};
\]
\item \(
    {\cal S}^+(I) = {\cal S}^+(I_1) \oplus {\cal S}^+(I_2),
\)
where $\oplus$ represents the $(\max, +)$-convolution defined by the following formula.
\[
    X \oplus Y = \{\text{$(w, p) \in X + Y$ : $(w, p)$ is not dominated in $X+Y$}\}.
\]
\end{itemize}

Using Chi et al.'s $\tilde{O}(n^{3/2})$-time algorithm for bounded monotone $(\min,+)$-convolution~\cite{CDXZ22} and Bringmann and Cassis's algorithm for approximating the $(\max,+)$-convolution~\cite[Lemma 28 in the full version]{BC22}, we can approximate the $(\max,+)$-convolution of two sets in $\tilde{O}((\frac{1}{\eps})^{3/2})$ time. Using the divide-and-conquer approach from~\cite{Chan18}, we can generalize it to multiple sets. We defer the proof of the following lemma to Appendix~\ref{app:conv}.

\begin{restatable}{lemma}{lemmerging}\label{lem:merging}
    Let $S_1, S_2, \ldots, S_m \subseteq (\{0\} \cup [A, B]) \times (\{0\} \cup [C, D])$ be sets of tuples with total size $\ell$. In $\tilde{O}(\ell + (\frac{1}{\eps})^{3/2}m)$ time\footnote{Here, the $\tilde{O}(\cdot)$ notation hides polylogarithmic factor not only in $n$ and $\frac{1}{\eps}$, but also in other parameters such as $\frac{B}{A}$ and $\frac{D}{C}$. These parameters will eventually be set to be polynomial in $n$ and $\frac{1}{\eps}$.} and with high probability, we can compute a set $\tilde{S}$ of size $\tilde{O}(\frac{1}{\eps})$ that approximates $S_1 \oplus S_2 \oplus \cdots \oplus S_m$ with factor $1 + O(\eps)$.  Moreover, $\tilde{S}$ is dominated by $S_1 + S_2 + \cdots + S_m$.
\end{restatable}

With Lemma~\ref{lem:merging}, we can approximate $\mathcal{S}^+(I) \cap ([0, t]\times [0, OPT])$ as follows. Partition $I$ into several groups, approximate each group independently, and merge their results using Lemma~\ref{lem:merging}.  As long as there are not too many groups, the merging time will be subquadratic. The following lemmas explain how the multiplicative factor and additive error accumulate during the computation. Their proofs are deferred to Appendix~\ref{app:err}.

\begin{restatable}{lemma}{lemapproxerra}
    \label{lem:approx-err-1}
    Let $S$, $S_1$, and $S_2$ be sets of tuples. 
    \begin{romanenumerate}
        \item If $S_1$ approximates $S$ with factor $1 + \eps_1$ and $S_2$ approximates $S_1$ with factor $1 + \eps_2$, then $S_2$ approximates $S$ with factor $(1 + \eps_1)(1 + \eps_2)$.

        \item If $S_1$ approximates $S$ with additive error $(\delta_{w1}, \delta_{p1})$ and $S_2$ approximates $S_1$ with factor $(\delta_{w2}, \delta_{p2})$, then $S_2$ approximates $S$ with additive error $(\delta_{w1} + \delta_{w2}, \delta_{p1}+ \delta_{p2})$.
    \end{romanenumerate}
\end{restatable}

\begin{restatable}{lemma}{lemapproxerrb}\label{lem:approx-err-2}
    Let $S_1$, $S_2$, $S'_1$, $S'_2$ be sets of tuples.
    \begin{romanenumerate}
        \item If $S'_1$ and $S'_2$ approximates $S_1$ and $S_2$ respectively with factor $1 + \eps$, then 
        $S'_1 \oplus S'_2$ and $S'_1 + S'_2$ approximate $S_1 \oplus S_2$ and $S_1 + S_2$ respectively with factor $1 + \eps$.

        \item If $S'_1$ and $S'_2$ approximates $S_1$ and $S_2$ with additive error $(\delta_{w1}, \delta_{p1})$ and $(\delta_{w2}, \delta_{p2})$, respectively, then 
        $S'_1 \oplus S'_2$ and $S'_1 + S'_2$ approximate $S_1 \oplus S_2$ and $S_1 + S_2$ respectively with additive error $(\delta_{w1} + \delta_{w2}, \delta_{p1}+ \delta_{p2})$.
    \end{romanenumerate}
\end{restatable}

In the rest of the paper, we may switch between the multiplicative factor and the additive error, although our goal is to bound the latter. It is helpful to keep the following observation in mind.
\begin{observation}
    If $S'$ approximates $S$ with factor $1 + \eps$, then $S'$ approximates each tuple $(w, p) \in S$ with additive error $(\eps w, \eps p)$.
\end{observation}

\section{Reducing the Problem}\label{sec:reduce}
With Lemma~\ref{lem:merging}, we can reduce the problem of weakly approximating Knapsack to the following problem. 
\begin{quote}
The reduced problem $\mathrm{RP}(\alpha)$: given a set $I$ of $n$ items from $[1,2] \times [1,2]$ and a real number $\alpha \in [\frac{1}{2}, \frac{1}{4\eps}]$, compute a set $S$ of size $\tilde{O}(\frac{1}{\eps})$ that approximates ${\cal S}^+(I) \cap ([0,\frac{1}{\alpha\eps}] \times [0, \frac{1}{\alpha\eps}])$ with additive error $(O(\frac{1}{\alpha}), O(\frac{1}{\alpha}))$. Moreover, $S$ should be dominated by $\mathcal{S}(I)$.
\end{quote}

\begin{restatable}{lemma}{lemreduce}
\label{lem:reduce}
    An $\tilde{O}(n + T(\frac{1}{\eps}))$-time algorithm for the reduced problem $\mathrm{RP}(\alpha)$ implies an $\tilde{O}(n + (\frac{1}{\eps})^{3/2} + T(\frac{1}{\eps}))$-time weak approximation scheme for Knapsack.
\end{restatable}

The approach to reducing the problem is similar to that in~\cite{CLMZ24bSTOCKnapsack}. The details of the proof are deferred to Appendix~\ref{app:reducing}.

Throughout, we write ${\cal S}^+(I) \cap ([0,\frac{1}{\alpha\eps}] \times [0, \frac{1}{\alpha\eps}])$ as $\mathcal{S}^+(I; \frac{1}{\alpha\eps})$ for short, and when the weight and profit both have an additive error of $\delta$, we simply say that additive error is $\delta$.

We also remark that when presenting algorithms for the reduced problem $\mathrm{RP}(\alpha)$, we will omit the requirement that the resulting set $S$ should be dominated by $\mathcal{S}(I)$ since it is guaranteed by the algorithms in a quite straightforward way. More specifically, in those algorithms (including Lemma~\ref{lem:merging}), rounding is the only place that may incur approximation errors, and whenever we round a tuple $(w,p)$, we always round $w$ up and round $p$ down. As a consequence, a set after rounding is always dominated by the original set. 

\section{Algorithmic Ingredients}\label{sec:tools}
We present a few results that will be used by our algorithm as subroutines.
\subsection{An Algorithm for Small Solution Size}
Consider the reduced problem $\mathrm{RP}(\alpha)$. Since items are from $[1,2] \times [1,2]$, it is easy to see that $|I'| \leq \frac{1}{\alpha\eps}$ for any $I' \subseteq I$ with $(w(I'), p(I')) \in {\cal S}^+(I; \frac{1}{\alpha\eps})$. When $\alpha$ is large, $|I'|$ is small, and such a case can be tackled using Bringmann's two-layer color-coding technique~\cite{Bri17}. Roughly speaking, we can divide $I$ into $k = \tilde{O}(\frac{1}{\alpha\eps})$ subsets $I_1, \ldots, I_k$ so that with high probability, $|I_j \cap I'| \leq 1$ for every $j$. Then it suffices to (approximately) compute $I_1 \oplus \cdots \oplus I_k$, and the resulting set will provide a good approximation for the tuple $(w(I'), p(I'))$.  Since this technique becomes quite standard nowadays, we defer the proof to Appendix~\ref{app:color-coding}.

\begin{restatable}{lemma}{lemalgforlargealpha}\label{lem:alg-for-large-alpha}
    Let $I$ be a set of items from $[1, 2]\times [1,2]$.  In $\tilde{O}(n + (\frac{1}{\eps})^{5/2}\frac{1}{\alpha})$ time and with high probability, we can compute a set $S$ of size $\tilde{O}(\frac{1}{\eps})$ that approximates ${\cal S}^+(I; \frac{1}{\alpha\eps})$ with factor $1 + O(\eps)$.
\end{restatable}

\subsection{An Efficiency-Based Proximity Bound}
A proximity bound basically states that the optimal solution of Knapsack must contain lots of high-efficiency items and very few low-efficiency items. Let $I = \{(w_1, p_1), \ldots, (w_n, p_n)\}$ be a set of items. Assume that the items are labeled in a non-increasing order of efficiencies. That is, $\frac{p_1}{w_1} \geq \cdots \geq \frac{p_n}{w_n}$. (The labeling may not be unique. We fix an arbitrary one in that case.) Given a capacity $w$, let $b$ be the minimum integer such that $w_1 + \cdots + w_b > w$. We say that the item $(w_b, p_b)$ is the breaking item with respect to $w$. (When no such $b$ exists, we simply let $b = n+1$ and define the breaking item to be $(+\infty, 0)$.) 

The following proximity bound is essentially the same as the one used in~\cite{Jin19,DJM23}. For completeness, we provide a proof in Appendix~\ref{app:proximity}.

\begin{restatable}{lemma}{lemproximity}\label{lem:proximity}
    Let $I$ be a set of items from $[1,2]\times [1,2]$. Let $(w^*, p^*) \in {\cal S}^+(I)$.   Let $I^* \subseteq I$ be a subset of items with $w(I^*) = w^*$ and $p(I^*) = p^*$.  Let $\rho_b$ be the efficiency of the breaking item with respect to $w^*$.  For any $I' \subseteq I$, 
    \begin{romanenumerate}
        \item if every item in $I'$ has efficiency at most $\rho_b - \Delta$ for some $\Delta > 0$, then 
        \[
            w(I' \cap I^*) \leq \frac{2}{\Delta};
        \]

        \item if every item in $I'$ has efficiency at least $\rho_b + \Delta$ for some $\Delta > 0$, then 
        \[
            w(I' \setminus  I^*) \leq \frac{2}{\Delta}.
        \]
    \end{romanenumerate}
\end{restatable}

\subsection{An Algorithm for Items with Similar Efficiencies}
When all the items have the same efficiency, Knapsack becomes Subset Sum, and there is an $\tilde{O}(n + \frac{1}{\eps})$-time weak approximation scheme for Subset Sum~\cite{CLMZ24cSTOCPartition}. The following lemma extends this result to the case where the item efficiencies differ by only a small amount.

\begin{lemma}\label{lem:alg-for-similar-eff}
    Let $I$ be a set of $n$ items from $[1,2]\times [0,\infty]$ whose efficiencies are in the range $[\rho, \rho + \Delta]$. In $\tilde{O}(n + \frac{1}{\eps} + (\frac{1}{\eps})^2 \frac{\Delta}{\rho})$ time and with high probability, we can compute a set of size $\tilde{O}(\frac{1}{\eps})$ that approximates ${\cal S}^+(I)$ with factor $1 + O(\eps)$.
\end{lemma}

We first consider the $(\max, +)$-convolution of two sets of tuples.
\begin{lemma}\label{lem:2d-fft}
    Let $S_1, S_2$ be two sets of tuples from $(\{0\}\cup[A, B]) \times [0, +\infty])$ with total size $\ell$. Suppose that the efficiencies of the tuples in $S_1 \cup S_2$ are in the range $\{0\} \cup [\rho, \rho + \Delta]$. In $\tilde{O}(\ell + \frac{1}{\eps} + (\frac{1}{\eps})^2\frac{\Delta}{\rho})$ time\footnote{Here, the $\tilde{O}(\cdot)$ notation hides polylogarithmic factor not only in $n$ and $\frac{1}{\eps}$, but also in $\frac{B}{A}$. These parameters will eventually be set to be polynomial in $n$ and $\frac{1}{\eps}$.}, we can compute a set $S$ of size $\tilde{O}(\frac{1}{\eps})$ that approximates $S_1 \oplus S_2$ with factor $1 + O(\eps)$. Moreover, $S$ is dominated by $S_1 + S_2$.
\end{lemma}
\begin{proof}
    It suffices to consider the case where $\rho = 1$. If $\rho \neq 1$, we can replace every tuple $(w, p)$ with $(w, \frac{p}{\rho})$, and replace $\Delta$ with $\frac{\Delta}{\rho}$. 
    
    Let $a \in [A, B]$.  We first describe how to approximate $(S_1 \oplus S_2)\cap ([0 ,a] \times [0, + \infty])$ with additive error $O(\eps a)$.

    Let $S'_1$ and $S'_2$ be the sets obtained from $S_1$ and $S_2$, respectively, by replacing every tuple $(w, p)$ with $(w, p - w)$. The range of item efficiencies of $S'_1$ and $S'_2$ becomes $[0, \Delta]$. We assume that every tuple $(w, p) \in S'_1 \cup S'_2$ has $w \leq a$ since the tuples with $w > a$ can be safely ignored. For every tuple $(w, p)$, we round $w$ up to and round $p$ down to the nearest multiple of $\eps a$. Denote the resulting sets by $S''_1$ and $S''_2$.  By scaling, we can assume that every tuple in $S''_1 \cup S''_2$ is a tuple of integers from $[0, \frac{1}{\eps}] \times [0, \frac{\Delta}{\eps}]$. Then we can compute $S = S''_1 + S''_2$ using 2D-FFT in $\tilde{O}(\frac{1}{\eps} + (\frac{1}{\eps})^2\Delta)$ time. It is easy to see that $S$ is of size $\tilde{O}((\frac{1}{\eps})^2\Delta)$, and is dominated by $S'_1 + S'_2$ (since we round $w$ up and round $p$ down). Moreover, for each $(w, p) \in S'_1 + S'_2$, there is a tuple $(w', p') \in S$ such that $w\leq w' \leq w + 2\eps a$ and $p - 2\eps a \leq p' \leq p$, which implies that $p + w - 4\eps a \leq p' + w' - 2\eps a \leq p + w$.
    Now consider the set $S'$ obtained from $S$ by replacing each tuple $(w, p)$ with $(w, p + w - 2\eps a)$. One can verify that $S'$ is dominated by $S_1 + S_2$ and approximates $S_1 + S_2$ with additive error $O(\eps a)$.  Since $S_1 \oplus S_2 \subseteq S_1 + S_2$, we have that $S'$ also approximates $S_1 \oplus S_2$ with additive error $O(\eps a)$. Then we remove all the dominated tuples from $S'$. This can be done in $O(|S'|\log |S'|) = \tilde{O}((\frac{1}{\eps})^2\Delta)$ time. After this, the size of $S'$ becomes $O(\frac{1}{\eps})$ because for each integer $w \in [0,\frac{1}{\eps}]$, there can be at most one tuple with weight $w$ in $S'$.

    To finish the proof, we repeat the above procedure for every $a \in [A, B]$ that is a power of $2$, and take the union of the $\log \frac{B}{A}$ resulting sets. Note that for a tuple $(w, p) \in S_1 + S_2$ with $w \in [a/2, a]$, we have that $p \geq \rho w = w \geq a/2$.  Therefore, the additive error $O(\eps a)$ implies an approximation factor of $1 + O(\eps)$. The running time increases by a factor of $\log \frac{B}{A}$.
\end{proof}

Now we are ready to prove Lemma~\ref{lem:alg-for-similar-eff}. 
\begin{proof}[Proof of Lemma~\ref{lem:alg-for-similar-eff}]
If $\Delta/\rho< \eps$, we can round down item profits so that the efficiencies of all items are $\rho$. This increases the approximation factor by at most $1 + \eps$. Now the problem degenerates to Subset Sum, and therefore, we can compute a set $S$ that approximates ${\cal S}^+(I)$ in $\tilde{O}(n + \frac{1}{\eps})$ time by the weak approximation scheme\footnote{The weak approximation scheme in~\cite{CLMZ24cSTOCPartition} actually returns a set $\tilde{S}$ that approximates ${\cal S}^+(I_j)$ with additive error $\eps t$. This can be easily modified to return a set that approximates ${\cal S}^+(I_j)$ with factor $1+\eps$ by repeating it for $t = 1,2,4,8,\ldots$, and taking the union of the resulting sets.}.

Now we assume $\Delta/\rho\geq\eps$. By rounding down item profits, we assume that the efficiencies of items are powers of $1 + \eps$. This increases the approximation factor by at most $1 + \eps$. So there are at most $m = \log_{1+ \eps} \frac{\Delta}{\rho} = O(\frac{\Delta}{\eps\rho})$ distinct efficiencies.  We divide the items into $m$ groups $I_1, \ldots, I_m$ by their efficiencies.  For each group $I_j$, since the items have the same efficiency, and again, we can compute a set $S_j$ that approximates ${\cal S}^+(I_j)$ in $\tilde{O}(|I_j| + \frac{1}{\eps})$ time by the weak approximation scheme
    for Subset Sum~\cite{CLMZ24cSTOCPartition}. The total time cost for computing all sets $S_j$ is $\tilde{O}(n + \frac{m}{\eps} + (\frac{1}{\eps})^2 \frac{\Delta}{\rho})=\tilde{O}(n + (\frac{1}{\eps})^2 \frac{\Delta}{\rho})$.

    Then we shall merge $S_1, \ldots, S_m$ using a divide-and-conquer approach. The tuples in $S_j$ have the same efficiency (except that the tuple $(0,0)$ has efficiency $0$), which is the efficiency of the items in $I_j$ . Let $\rho_j$ be this efficiency. Without loss of generality, assume that $\rho_1 < \cdots < \rho_m$. We recursively (and approximately) compute $S_1 \oplus \cdots \oplus S_{m/2}$ and $S_{m/2 + 1} \oplus \cdots \oplus S_m$, and then merge the resulting two sets via Lemma~\ref{lem:2d-fft}. The recursion tree has $\log m$ levels. One can verify that the total merging time in each level is $\tilde{O}(\frac{m}{\eps} + (\frac{1}{\eps})^2 \frac{\rho_m - \rho_1}{\rho_1})$, which is bounded by $\tilde{O}((\frac{1}{\eps})^2 \frac{\Delta}{\rho})$ since $\rho \leq \rho_j \leq \rho + \Delta$. The approximation factor is $(1 + \eps)^{\log m} = 1 + O(\eps \log m)$. We can adjust $\eps$ by a factor of $\log m$, which increases the running time by a polylogarithmic factor.
\end{proof}

The Lemma~\ref{lem:alg-for-similar-eff} immediately implies the following.
\begin{corollary}\label{coro:approx-add}
    Let $I$ be a set of $n$ items from $[1, 2] \times [0, +\infty]$ whose efficiencies are in the range $[\rho, \rho + \Delta]$. In $\tilde{O}(n + \frac{1}{\eps} + (\frac{1}{\eps})^2 \frac{\Delta}{\rho})$ time and with high probability, we can compute a set of size $\tilde{O}(\frac{1}{\eps})$ that approximates each tuple $(w,p) \in {\cal S}^+(I)$ with additive error $(\eps w, \eps p)$. 
\end{corollary}

In the above corollary, the additive error is proportional to the value of the entry of the tuple: a large entry means a large additive error. We can also achieve the contrary and approximate in a way so that a large entry means a small additive error. Basically, this can be done by (approximately) solving a problem symmetric to Knapsack: instead of determining which items should be selected, we determine which items should not be selected. 

\begin{corollary}\label{coro:approx-del}
    Let $I$ be a set of $n$ items from $[1, 2] \times [0, +\infty]$ whose efficiencies are in the range $[\rho, \rho + \Delta]$. In $\tilde{O}(n + \frac{1}{\eps} + (\frac{1}{\eps})^2 \frac{\Delta}{\rho})$ time and with high probability, we can compute a set of size $\tilde{O}(\frac{1}{\eps})$ that approximates ${\cal S}^+(I)$ with additive error $(\eps(w(I) - w), \eps(p(I) - p))$.
\end{corollary}
\begin{proof}
    Define
    \[
        {\cal S}^-(I) = \{\text{$(w, p) \in {\cal S}(I)$ : $(w,p)$ does not dominate any other tuple in ${\cal S}(I)$}\}.
    \]
    
    From the perspective of Knapsack, every tuple in ${\cal S}^+(I)$ corresponds to an optimal solution of Knapsack, while every tuple $(w,p)$ in ${\cal S}^-(I)$ corresponds to an optimal solution of the symmetric problem that seeks for the minimum total profit that can be achieved using capacity at least $w$.  It is easy to observe that
    \[
        {\cal S}^+(I) = \{(w(I) - w, p(I) - p) : (w, p) \in {\cal S}^-(I)\}.
    \]
    Therefore, to approximate ${\cal S}^+(I)$, it suffices to approximate ${\cal S}^-(I)\}$.

    Geometrically, ${\cal S}^+(I)$ represents the upper envelope of ${\cal S}(I)$, while ${\cal S}^-(I)$ represents the lower envelope.  But we can make ${\cal S}^-(I)$ be the upper envelope by exchanging the $x$-coordinate and $y$-coordinate. Then we can approximate ${\cal S}^-(I)$ using Corollary~\ref{coro:approx-add}.

    More precisely, let $I' = \{(p, w) : (w, p) \in I\}$. The efficiencies of the items in $I'$ are in the range $[\rho', \rho' + \Delta']$, where $\rho' = \frac{1}{\rho + \Delta}$ and $\rho' + \Delta' = \frac{1}{\rho}$. One can verify that $\frac{\Delta'}{\rho'} = \frac{\Delta}{\rho}$. Then we can compute a set $S$ of size $\tilde{O}(\frac{1}{\eps})$ that approximates each tuple $(p,w) \in {\cal S}^+(I')$  with additive error $(\eps p, \eps w)$ via Corollary~\ref{coro:approx-add} in $\tilde{O}(n + \frac{1}{\eps}+ (\frac{1}{\eps})^2\frac{\Delta}{\rho})$ time. Let 
    \[
        S'  = \{ (w(I) - w, p(I) - p) : (p, w)\in S\}.
    \]
    One can verify that $S'$ approximates each tuple $(w, p) \in {\cal S}^+(I)$ with additive error $(\eps(w(I) - w), \eps(p(I) - p))$.
\end{proof}

\section{An Algorithm with Exponent 11/6}\label{sec:weak}
To illustrate our main idea, we first present a simpler algorithm that runs in $\tilde{O}(n + (\frac{1}{\eps})^{11/6})$ time.  When $\alpha \geq (\frac{1}{\eps})^{2/3}$, the algorithm in Lemma~\ref{lem:alg-for-large-alpha} already runs in $\tilde{O}(n + (\frac{1}{\eps})^{11/6})$ time. It remains to consider the case where $\frac{1}{2}\leq \alpha \leq (\frac{1}{\eps})^{2/3}$. In the rest of this section, we assume that $I = \{(w_1, p_1), \ldots, (w_n, p_n)\}$ and that $\frac{p_1}{w_1} \geq \cdots \geq \frac{p_n}{w_n}$.

Let $\tau$ be a parameter that will be specified later.  We first partition $I$ into groups as follows. For the items $(w_1, p_1), \ldots, (w_n, p_n)$ in $I$, we bundle every $\tau$ items as a group until we obtain $\lceil\frac{1}{\alpha\eps\tau}\rceil + 1$ groups. For the remaining items, say $\{(w_i, p_i), \ldots, (w_n, p_n)\}$, we further divide them into two groups $\{(w_i, p_i), \ldots, (w_{i'}, p_{i'})\}$ and $\{(w_{i'+ 1}, p_{i'+1}), \ldots, (w_n, p_n)\}$, where $i'$ is the maximum index such that $\frac{p_i}{w_i} - \frac{p_{i'}}{w_{i'}} \leq \frac{1}{\tau}$. At the end, we obtain $m = \lceil\frac{1}{\alpha\eps\tau}\rceil + 3$ groups $I_1, I_2, \ldots, I_m$. For simplicity, we assume that none of these groups is empty. This assumption is without loss of generality, since empty groups only make things easier.

For each group $I_j$, we define $\Delta_j$ to be the maximum difference between the efficiencies of the items in $I_j$. That is,
\[
    \Delta_j = \max\left\{\frac{p}{w} - \frac{p'}{w'} : (w, p), (w', p') \in I_j\right\}. 
\]
By the way we construct the groups, we have that $\sum_j \Delta_j \leq \frac{3}{2}$.  

We say a group $I_j$ is good if $\Delta_j \leq \frac{1}{\tau}$, and bad otherwise. 

Note that the efficiency of any item in $I$ is within the range $[\frac{1}{2}, 2]$. For each good group $I_j$, its item efficiencies are in the range $[\rho, \rho + \frac{1}{\tau}]$ for some constant $\rho \in [\frac{1}{2},2]$, so we can compute a set $S_j$ of $\tilde{O}(\frac{1}{\eps})$ that approximates each tuple $(w, p) \in \mathcal{S}^+(I_j)$ with additive error $(\eps w, \eps p)$ via Corollary~\ref{coro:approx-add} in $O(|I_j| + \frac{1}{\eps} + (\frac{1}{\eps})^2\frac{1}{\tau})$ time. Since the items are from $[1,2]\times [1,2]$, for any tuple in $\mathcal{S}^+(I_j)$, its weight and profit are of the same magnitude, so are the corresponding additive errors. We can say that $S_j$ approximates each tuple $(w, p) \in \mathcal{S}^+(I_j)$ with additive error $O(\eps w)$.

For bad groups, we shall approximate them less accurately: we shall use $\frac{1}{\alpha\tau}$ instead of $\eps$ as the accuracy parameter. (Later we will show that we can still obtain a good bound on the total additive error using the proximity bound in Lemma~\ref{lem:proximity}.)  We compute a set $S'_j$ of size $\tilde{O}(\alpha\tau)$ that approximates each tuple $(w, p) \in \mathcal{S}^+(I_j)$ with additive error $(\frac{w}{\alpha \tau}, \frac{p}{\alpha \tau})$ via Corollary~\ref{coro:approx-add} in $O(|I_j| + \alpha\tau + \alpha^2\tau^2\Delta_j)$ time. We also compute a set $S''_j$ of size $\tilde{O}(\alpha\tau)$ that approximates each $(w, p) \in \mathcal{S}^+(I_j)$ with additive error $(\frac{w(I)-w}{\alpha \tau}, \frac{p(I) - p}{\alpha \tau})$ via Corollary~\ref{coro:approx-del} in $O(|I_j| + \alpha\tau + \alpha^2\tau^2\Delta_j)$ time. Again, since the weight and the profit are of the same magnitude, we can say that $S'_j$ approximates each tuple $(w, p) \in \mathcal{S}^+(I_j)$ with additive error $O(\frac{1}{\alpha \tau}\cdot w)$, and that $S''_j$ approximates each tuple $(w, p) \in \mathcal{S}^+(I_j)$ with additive error $O(\frac{1}{\alpha \tau}\cdot(w(I) - w))$. Now consider the set $S_j = S'_j \cup S''_j$. One can verify that $S_j$ approximates each tuple $(w, p) \in \mathcal{S}^+(I_j)$ with additive error $O(\frac{1}{\alpha \tau}\cdot \min\{w, (w(I) - w)\})$.

The last step is to compute a set $S$ of size $\tilde{O}(\frac{1}{\eps})$ that approximates $S_1 \oplus \cdots \oplus S_m$ with factor $1+O(\eps)$ via Lemma~\ref{lem:merging}. Note that each set $S_j$ is of size $\tilde{O}(\frac{1}{\eps})$ or $\tilde{O}(\alpha\tau)$. This step takes time $\tilde{O}(\frac{m}{\eps} + m\alpha\tau + (\frac{1}{\eps})^{3/2}m)$.

\paragraph*{Bounding Total Time Cost}
Partitioning the item into $I_1, \ldots, I_m$ can be done by sorting and scanning, and therefore, it takes only $O(n\log n + m)$ time. Each good group costs $O(|I_j| + \frac{1}{\eps} + (\frac{1}{\eps})^2\frac{1}{\tau})$ time. Since there are at most $m$ good groups, the good groups cost  total time
\[
    \tilde{O}\left(n + \frac{m}{\eps} + (\frac{1}{\eps})^2\frac{m}{\tau}\right).
\]
Each bad group $I_j$ costs $O(|I_j| + \alpha\tau + \alpha^2\tau^2\Delta_j)$ time, so the total time cost for bad groups is
\[
    \tilde{O}(n + m\alpha\tau + \alpha^2\tau^2\sum_j\Delta_j) \leq \tilde{O}(n + m\alpha\tau + \alpha^2\tau^2).
\]
The inequality is due to that $\sum_j\Delta_j \leq \frac{3}{2}$. Merging all the sets $S_j$ via Lemma~\ref{lem:merging} costs $\tilde{O}(\frac{m}{\eps} + m\alpha\tau + (\frac{1}{\eps})^{3/2}m)$ time. Taking the sum of all these time costs, we have that the running time of the algorithm is
\[
    \tilde{O}\left(n  + \frac{m}{\eps} + m\alpha\tau + (\frac{1}{\eps})^2\frac{m}{\tau} + \alpha^2\tau^2 + (\frac{1}{\eps})^{3/2}m\right).
\]
Recall that $m = \lceil \frac{1}{\alpha\eps\tau} \rceil$ and that $\frac{1}{2} \leq \alpha \leq (\frac{1}{\eps})^{2/3}$. Set $\tau = \lceil(\frac{1}{\eps})^\frac{11}{12}\frac{1}{\alpha}\rceil$. One can verify that the running time is $\tilde{O}(n + (\frac{1}{\eps})^{11/6})$.

\paragraph*{Bounding Additive Error}
To prove the correctness of the algorithm, it suffices to show that the set $S$ returned by the algorithm approximates every tuple in ${\cal S}^+(I; \frac{1}{\alpha\eps})$ with additive error $O(\frac{1}{\alpha})$. 

Let $(w^*, p^*)$ be an arbitrary tuple in ${\cal S}^+(I; \frac{1}{\alpha\eps})$.  Let $(w_b, p_b)$ be the breaking item with respect to $w^*$. Let $\rho_b = \frac{p_b}{w_b}$ be the efficiency of the breaking item. We first show that the proximity bound can be applied to all but four bad groups.

\begin{lemma}\label{lem:bound-median}
    At most four bad groups may contain an item whose efficiency is within the range $[\rho_b - \frac{1}{\tau}, \rho_b + \frac{1}{\tau}]$. Moreover, each of these groups is of size $\tau$.
\end{lemma}
\begin{proof}
    We first show that there are at most four bad groups each containing an item whose efficiency is within the range $[\rho_b - \frac{1}{\tau}, \rho_b + \frac{1}{\tau}]$.  Suppose, for the sake of contradiction, that there are five bad groups each containing an item whose efficiency is in the range $[\rho_b - \frac{1}{\tau}, \rho_b + \frac{1}{\tau}]$.  Let $i_1, i_2, i_3, i_4, i_5$ be these five items. Without loss of generality, assume that $\frac{p_{i_1}}{w_{i_1}} \geq \cdots \geq \frac{p_{i_5}}{w_{i_5}}$. Since they belong to five bad groups, and within each bad group, the maximum and the minimum item efficiency differ by at least $\frac{1}{\tau}$, we have that $\frac{p_{i_1}}{w_{i_1}} - \frac{p_{i_5}}{w_{i_5}} \geq \frac{3}{\tau}$. But this is impossible since the efficiencies of $i_1$ and $i_5$ in the range $[\rho_b - \frac{1}{\tau}, \rho_b + \frac{1}{\tau}]$ whose length is $\frac{2}{\tau}$.

    Note that by our construction, $I_{m-1}$ is a good group, so it cannot be one of these four bad groups. 
    Since the first $m-2$ groups are of size $\tau$, to finish the proof, it suffices to show that $I_{m}$ cannot have an item whose efficiency is within the range $[\rho_b - \frac{1}{\tau}, \rho_b + \frac{1}{\tau}]$.  Since $(w^*, p^*) \in {\cal S}^+(I; \frac{1}{\alpha\eps})$, we have that $w^* \leq \frac{1}{\alpha\eps}$. Since every item is from $[1,2]\times[1,2]$, we have $b \leq \frac{1}{\alpha\eps} + 1$. Recall that $m = \lceil\frac{1}{\alpha\eps\tau}\rceil + 3$ and that the first $m-2$ groups each has $\tau$ items. So it must be that $(w_b, p_b) \in I_j$ for some $j \leq m-2$. Let $i'$ be the first item (the item with the highest efficiency) in $I_{m-1}$. We have $ \frac{p_{i'}}{w_{i'}} \leq \frac{p_b}{w_b}$. Let $i''$ be the first item in $I_m$.  By the construction of $I_{m-1}$, we have that 
    \(
        \frac{p_{i''}}{w_{i''}} < \frac{p_{i'}}{w_{i'}} -  \frac{1}{\tau}.
    \)
    Therefore, 
    \(
        \frac{p_{i''}}{w_{i''}} < \frac{p_{b}}{w_{b}} -  \frac{1}{\tau}.
    \)
    This implies that every item in $I_m$ has its efficiency strictly less than $\frac{p_{b}}{w_{b}} -  \frac{1}{\tau}$.
\end{proof}

Now we are ready to bound the total additive error.
\begin{lemma}\label{lem:total-err}
    $S_1 \oplus S_2 \oplus \cdots \oplus S_m$ approximates $(w^*, p^*)$ with additive error $O(\frac{1}{\alpha})$.
\end{lemma}
\begin{proof}
    Since $I = I_1 \cup \cdots \cup I_m$ and $(w^*, p^*) \in \mathcal{S}^+(I)$, we have that there is a tuple $(w^*_j, p^*_j) \in \mathcal{S}^+(I_j)$ for each $j$ such that 
    \[
        (w^*, p^*) = (w^*_1, p^*_1) + \cdots + (w^*_m, p^*_m).
    \]
    Let $J_{\mathit{good}}$ be the indices of good groups. Let $J_{\mathit bad}$ be the set of indices of bad groups. Recall that for each $j \in J_{\mathit good}$, $S_j$ approximates $(w^*_j, p^*_j)$ with additive error $O(\eps w^*_j)$, and that for each $j \in J_{\mathit bad}$, the set $S_j$ approximates $(w^*_j, p^*_j)$ with additive error $O(\frac{1}{\alpha\tau} \min\{w^*_j, w(I_j) - w^*_j\})$.  Therefore, $S_1 \oplus S_2 \oplus \cdots \oplus S_m$ approximates $(w^*, p^*)$ with additive error
    \[
        O(\eps\cdot \sum_{j \in J_{\mathit{good}}} w^*_j + \frac{1}{\alpha\tau}\cdot\sum_{j \in J_{\mathit{bad}}} \min\{w^*_j, w(I_j) - w^*_j\}).
    \]
    It is easy to see that
    \begin{equation}\label{eq:bound-good-groups}
         \sum_{j \in J_{\mathit{good}}} w^*_j \leq w^* \leq \frac{1}{\alpha\eps}.
    \end{equation}
    It remains to bound $\sum_{j \in J_{\mathit{bad}}} \min\{w^*_j, w(I_j) - w^*_j\}$. We further divide $J_{\mathit{bad}}$ into $J^H_{\mathit{bad}} \cup J^M_{\mathit{bad}} \cup J^L_{\mathit{bad}}$ so that 
    \begin{itemize}
        \item For $j \in J^H_{\mathit{bad}}$, all items in $I_j$ have efficiency at least $\rho_b + \frac{1}{\tau}$

        \item For $j \in J^M_{\mathit{bad}}$, some items in $I_j$ has efficiency in the range $[\rho_b - \frac{1}{\tau}, \rho_b + \frac{1}{\tau}]$.

        \item For $j \in J^L_{\mathit{bad}}$, all items in $I_j$ have efficiency at most $\rho_b - \frac{1}{\tau}$.
    \end{itemize}
    The proximity Lemma~\ref{lem:proximity} implies that 
    \begin{equation}\label{eq:bound-high-and-low-groups}
        \sum_{j \in J^H_{\mathit{bad}}} (w(I_j) - w^*_j) \leq 2\tau \qquad\text{and}\qquad \sum_{j \in J^L_{\mathit{bad}}} w^*_j \leq 2\tau.
    \end{equation}
    Lemma~\ref{lem:bound-median} implies that $|J^M_{\mathit{bad}}| \leq 4$ and $I_j \leq \tau$ for each $j \in J^M_{\mathit{bad}}$. Recall that the item weights are at most $2$. We have
    \begin{equation}\label{eq:bound-median-groups}
        \sum_{j \in J^M_{\mathit{bad}}} w^*_j \leq 4 \cdot \tau \cdot 2 = 8\tau.
    \end{equation}
    In view of \eqref{eq:bound-high-and-low-groups} and \eqref{eq:bound-median-groups}, we have
    \begin{equation}\label{eq:bound-bad-groups}
        \sum_{j \in J_{\mathit{bad}}} \min\{w^*_j, w(I_j) - w^*_j\} \leq \sum_{j \in J^H_{\mathit{bad}}} (w(I_j) - w^*_j) + \sum_{j \in J^L_{\mathit{bad}}} w^*_j + \sum_{j \in J^M_{\mathit{bad}}} w^*_j \leq  12\tau.
    \end{equation}
    In view of \eqref{eq:bound-good-groups} and \eqref{eq:bound-bad-groups},  we have $S_1 \oplus S_2 \oplus \cdots \oplus S_m$ approximates $(w^*, p^*)$ with additive error
    \[
        O(\eps\cdot \sum_{j \in J_{\mathit{good}}} w^*_j + \frac{1}{\alpha\tau}\cdot\sum_{j \in J_{\mathit{bad}}} \min\{w^*_j, w(I_j) - w^*_j\}) \leq O(\eps \cdot \frac{1}{\alpha\eps} + \frac{1}{\alpha\tau} \cdot 12\tau) = O(\frac{1}{\alpha}). \qedhere
    \]
\end{proof}

The above lemma immediately implies the following since $S$ approximates $S_0 \oplus S_1 \oplus \cdots \oplus S_m$ with factor $1 + O(\eps)$.
\begin{corollary}
    $S$ approximates $(w^*, p^*)$ with additive error $O(\frac{1}{\alpha})$.
\end{corollary}

We summarize this section by the following lemma.
\begin{lemma}
   There is an $\tilde{O}({n+(\frac{1}{\eps})^{11/6}})$-time algorithm for the reduced problem, which is randomized and succeeds with high probability.
\end{lemma}

\section{An Improved Algorithm with Exponent 7/4}\label{sec:strong}
Due to the space limit, we present only the most crucial part of the algorithm and the analysis. A full version of this section can be found in Appendix~\ref{app:strong}.

In our $\tilde{O}(n + (\frac{1}{\eps})^{11/6})$-time algorithm, only the bad groups benefit from the proximity bound, which allows us to approximate with larger factor $1 + \frac{1}{\alpha\tau}$, while for the good groups, we simply approximate with factor $1 + \eps$. To further improve the running time, we shall partition the items in a way that all the groups can benefit from the proximity bound.

We assume that $\frac{1}{2}\leq \alpha \leq (\frac{1}{\eps})^{3/4}$ since the algorithm in Lemma~\ref{lem:alg-for-large-alpha} already runs in $\tilde{O}(n + (\frac{1}{\eps})^{7/4})$ time for all $\alpha \geq (\frac{1}{\eps})^{3/4}$. We assume that $I = \{(w_1, p_1), \ldots, (w_n, p_n)\}$ and that $\frac{p_1}{w_1} \geq \cdots \geq \frac{p_n}{w_n}$.

We first partition $I$ into the head group $I_{\mathrm{head}}$ and the tail group $I_{\mathrm{tail}}$ where $I_{\mathrm{head}}$ contains the first $\lceil \frac{1}{\alpha\eps}\rceil + 1 $ items and $I_{\mathrm{tail}}$ contains the rest of the items. 
Below we only give an algorithm for $I_{\mathrm{head}}$, and analyze its additive error. The remaining analysis and the algorithm for $I_{\mathrm{tail}}$ can be found in the full version (Appendix~\ref{app:strong}).

\subsection{Approximating Head Group}\label{sec:strong-head-incomp}
Let $n' = \lceil \frac{1}{\alpha\eps}\rceil + 1$.  We shall show that in $\tilde{O}((\frac{1}{\eps})^{7/4})$ time, we can compute a set $S$ of size $\tilde{O}(\frac{1}{\eps})$ that approximate $\mathcal{S}^+(I_{\mathrm{head}})$ with additive error $O(\frac{1}{\alpha})$.

Let $\tau$ be a parameter that will be specified later. Roughly speaking, we will further partition $I_{\mathrm{head}}$ into good groups and bad groups. The bad groups are the same as before: a bad group is a group of $\tau$ items whose efficiencies differ by at least $\frac{1}{\tau}$. For good groups, we will make them large than before: when we obtain a good group $I'$ of size $\tau$, we will keep adding items to $I'$ until the difference between the maximum and the minimum item efficiencies exceeds $\frac{1}{|I'|}$ with the next item. A more precise description of the partition process is given below. 

We partition $I_{\mathrm{head}}$ into $I_1, \ldots, I_m$ as follows. Initially, $j = 1$ and $k = 1$, and the remaining items are $\{(w_k, p_k), \ldots, (w_{n'}, p_{n'})\}$. Let $k'$ be the minimum integer such that
\[
    (k' - k) \cdot (\frac{p_k}{w_k} - \frac{p_{k'}}{w_{k'}}) > 1.
\]
If such $k'$ does not exist, we set $I_j = \{(w_k, p_k), \ldots, (w_{n'}, p_{n'})\}$, set $\Delta_j = \Delta'_j = \frac{p_k}{w_k} - \frac{p_{n'}}{w_{n'}}$, set $m: = j$, and finish. Assume that $k'$ exists. 
\begin{itemize}
    \item If $k' - k \geq \tau$, we set $I_j = \{(w_k, p_k), \ldots, (w_{k' - 1}, p_{k'-1})\}$, set $\Delta_j = \frac{p_k}{w_k} - \frac{p_{k'-1}}{w_{k'-1}}$, set $\Delta'_j = \frac{p_k}{w_k} - \frac{p_{k'}}{w_{k'}}$, and proceed with $j := j + 1$ and $k: = k'$. In this case, we say $I_j$ is a good group. 

    \item Otherwise, we let $I_j = \{(w_k, p_k), \ldots, (w_{k + \tau - 1}, p_{k + \tau - 1})\}$, set $\Delta_j = \frac{p_k}{w_k} - \frac{p_{k + \tau - 1}}{w_{k+\tau - 1}}$, set $\Delta'_j = \frac{p_k}{w_k} - \frac{p_{k + \tau}}{w_{k+\tau}}$, and proceed with $j: = j + 1$ and $k:= k + \tau$. In this case, we say that $I_j$ is a bad group. 
\end{itemize}
(We remark that $\Delta_j$ and $\Delta'_j$ are not required by the algorithm. They are maintained only for the purpose of analysis. $\Delta_j$ is the actual difference between the maximum and minimum item efficiencies in $I_j$. For technical reasons, we also need $\Delta'_j$. Basically, compared to $\Delta_j$, the gap $\Delta'_j$ also includes the efficiency gap between the minimum efficiency in $I_j$ and the maximum efficiency in $I_{j+1}$. We will use $\Delta_j$ to bound the time cost of $I_j$, and use $\Delta'_j$ to create an efficiency gap for other groups.)

We have the following observations.
\begin{observation}\label{obs:group-property-incomp}
The groups $I_1, \ldots, I_m$ satisfy the following properties.
\begin{romanenumerate}
    \item $m \leq \frac{1}{\alpha\eps\tau} + 1$

    \item $\sum_{j=1}^{m}|I_j| = |I_{\mathrm{head}}| = \frac{1}{\alpha\eps} + 1$ and $\sum_{j=1}^{m}\Delta_j \leq \sum_{j=1}^{m}\Delta'_j \leq \frac{3}{2}$.

    \item For every group $I_j$, the efficiencies of the items in $I_j$ differ by at most $\Delta_j$.

    \item For every good group $I_j$, we have $|I_j|\Delta_j \leq 1$ and $|I_j|\Delta'_j \geq \frac{1}{2}$.

    \item For every bad group $I_j$, we have $|I_j|\Delta'_j \geq |I_j|\Delta_j > 1$.

    \item For the last group $I_m$, we have $|I_m|\Delta_m \leq 1$.
\end{romanenumerate}
\end{observation}

For each group, we shall approximate it in exactly the same way as we did for the bad groups in the $\tilde{O}(n + (\frac{1}{\eps})^{11/6})$-time algorithm, except that we shall use $\frac{1}{\alpha|I_j|}$ as the accuracy parameter. More specifically, we use Corollary~\ref{coro:approx-add} and Corollary~\ref{coro:approx-del} to compute a set $S_j$ of size $\tilde{O}(\alpha|I_j|)$ that approximate each tuple $(w, p) \in \mathcal{S}^+(I_j)$ with additive error $\frac{1}{\alpha|I_j|} \cdot \min\{w, w(I) - w\}$. The time cost for $I_j$ is $\tilde{O}(|I_j| + \alpha|I_j| + \alpha^2|I_j|^2\Delta_j)$.

Then we compute a set $S_{\mathrm{head}}$ of size $\tilde{O}(\frac{1}{\eps})$ that approximates $S_1 \oplus \cdots \oplus S_m$ with factor $1+O(\eps)$ via Lemma~\ref{lem:merging}. 

\paragraph*{Bound Additive Error}
We shall show that $S_{\mathrm{head}}$ approximates every tuple $(w^*, p^*) \in \mathcal{S}^+(I_{\mathrm{head}})$ with additive error $\tilde{O}(\frac{1}{\alpha})$. 

Let $(w^*, p^*)$ be an arbitrary tuple in $\mathcal{S}^+(I_{\mathrm{head}})$. Let $(w_b, p_b)$ be the breaking item with respect to $w^*$. Let $\rho_b = \frac{p_b}{w_b}$ be the efficiency of the breaking item. 

Let $I_{j'}$ be the group containing the breaking item $(w_b, p_b)$. To apply the proximity bound, we will show that the groups can be divided into  $O(\log \frac{1}{\eps})$ collections so that for each collection $\{I_{j_{k} + 1}, I_{j_{k} + 2}, \ldots, I_{j_{k+1}-1} \}$, the efficiency gap between $\rho_b$ and the groups in this collection is at least the maximum $\Delta'_j$ of these groups. To do this,
we need the following auxiliary lemma. Its proof is deferred to the full version.

\begin{lemma}\label{lem:cluster-incomp}
    Let $\Delta_1, \ldots, \Delta_n$ be a sequence of positive real numbers. Let $\Delta_{\min}$ and $\Delta_{\max}$ be the minimum and maximum numbers in the sequence. There exists $h = O(\log \frac{\Delta_{\max}}{\Delta_{\min}})$ indices $1 = j_1 < j_2 < \cdots < j_h = n$ such that for any $k \in \{1, \ldots, h-1\}$, we have that
    \begin{equation}\label{eq:cluster-original-incomp}
            \max\{\Delta_j : j_k < j < j_{k+1}\} \leq \sum\{\Delta_j : j_{k+1} \leq j \leq n\},
    \end{equation}
    where the maximum of an empty set is defined to be $0$.
\end{lemma}

\begin{lemma}
    $S_1 \oplus \cdots \oplus S_m$ approximates $(w^*, p^*)$ with additive error $\tilde{O}(\frac{1}{\alpha})$.
\end{lemma}
\begin{proof}
    Since $I_{\mathrm{head}} = I_1 \cup \cdots \cup I_m$ and $(w^*, p^*) \in \mathcal{S}^+(I_{\mathrm{head}})$, we have that there is a tuple $(w^*_j, p^*_j) \in \mathcal{S}^+(I_j)$ for each $j$ such that 
    \(
        (w^*, p^*) = (w^*_1, p^*_1) + \cdots + (w^*_m, p^*_m).
    \)
    Recall that $S_j$ approximates $(w^*_j, p^*_j)$ with the additive error $\delta_j = \frac{1}{\alpha|I_j|}\cdot\min\{w^*_j,w(I_j)-w^*_j\}$. It suffices to bound $\sum_j \delta_j$. Let $I_{j'}$ be the group containing the breaking item $(w_b, p_b)$. Since the weights of the items are in $[1,2]$. It is easy to see that 
    \(
        \delta_{j} \leq \frac{1}{\alpha|I_{j}|}w^*_{j} \leq \frac{1}{\alpha}
    \)
    for any $j$.
    Therefore, $\delta_{j'} + \delta_{m} \leq \frac{2}{\alpha}$. It remains to consider $\sum_{j = 1}^{j'-1} \delta_j$ and $\sum_{j = j'}^{m-1} \delta_j$. We shall bound them using Lemma~\ref{lem:cluster-incomp} and the proximity bound (Lemma~\ref{lem:proximity}).

    Consider $\Delta'_1, \Delta'_2, \ldots, \Delta'_{j'-1}$. We shall apply Lemma~\ref{lem:cluster-incomp} to them. Let $\Delta'_{\max} = \max_{j=1}^{j'-1}\Delta'_j$ and let $\Delta'_{\min} = \min_{j=1}^{j'-1}\Delta'_j$.  Note that $\Delta'_{\max} \leq 3/2$ since the items are from $[1,2] \times [1,2]$, and that \(
    \Delta'_{\min} \geq 1/(2\sum_j |I_j|) \geq \frac{\alpha\eps}{2} \geq \frac{\eps}{4}
    \) due to Observation~\ref{obs:group-property-incomp}(ii), (iv), and (v). Therefore, $\log(\frac{\Delta'_{\max}}{\Delta'_{\min}}) = O(\log \frac{1}{\eps})$.  By Lemma~\ref{lem:cluster-incomp}, there exists $h = O(\log \frac{1}{\eps})$ indices $1 = j_1 < j_2 < \cdots < j_h = j' - 1$ such that for any $k \in \{1, \ldots, h-1\}$, we have that
    \begin{equation}\label{eq:cluster-incomp}
            \max\{\Delta'_j : j_k < j < j_{k+1}\} \leq \sum\{\Delta'_j : j_{k+1} \leq j \leq j'-1\}.
   \end{equation}
   Fix some $k$. Consider the groups $I_j$ with $j_k < j < j_{k+1}$. Let $\Delta' = \max\{\Delta'_j : j_k < j < j_{k+1}\}$. The inequality~\eqref{eq:cluster-incomp} implies that the items in these groups $I_j$ have efficiencies at least $\rho_b + \Delta'$. By Lemma~\ref{lem:proximity}, we have 
   \begin{equation}\label{eq:proximity-bound-b-incomp}
        \sum_{j_k < j < j_{k+1}} (w(I_j) - w^*_j) \leq \frac{2}{\Delta'}.  
   \end{equation}

   Also note that for each $I_j$ with $j_k < j < j_{k+1}$, we have $\frac{1}{|I_j|} \leq 2\Delta'_j \leq 2\Delta'$ (due to Observation~\ref{obs:group-property-incomp}(iv) and (v)). 
   Therefore, 
   \[
        \sum_{j_k < j < j_{k+1}} \delta_j \leq \sum_{j_k < j < j_{k+1}} \frac{w(I_j) - w^*_j}{\alpha|I_j|} \leq \sum_{j_k < j < j_{k+1}} \frac{2\Delta'(w(I_j) - w^*_j)}{\alpha} \leq \frac{1}{\alpha}.
   \]
   The last inequality is due to \eqref{eq:proximity-bound-b-incomp}.
   Recall that $\delta_j \leq \frac{1}{\alpha}$ for any $j$. We have the following.
   \[
        \sum_{j=1}^{j'-1} \delta_j  = \sum_{k=1}^h \delta_{j_k}  + \sum_{k=1}^{h-1}\sum_{j_k < j < j_{k+1}} \delta_j = O(\frac{h}{\alpha}) = \tilde{O}(\frac{1}{\alpha}).
   \]

   In a symmetric way, we can show that 
   \(
       \sum_{j=j'+ 1}^{m-1} \delta_j  \leq \tilde{O}(\frac{1}{\alpha}).
   \) 
\end{proof}

The above lemma immediately implies the following since $S_{\mathrm{head}}$ approximates $ S_1 \oplus \cdots \oplus S_m$ with factor $1 + O(\eps)$.
\begin{corollary}
    $S_{\mathrm{head}}$ approximates $(w^*, p^*)$ with additive error $\tilde{O}(\frac{1}{\alpha})$.
\end{corollary}

The remaining analysis and the algorithm for $I_{\mathrm{tail}}$ can be found in the full version (Appendix~\ref{app:strong}). We summarize this with the following lemma.
\begin{lemma}
   There is an $\tilde{O}({n+(\frac{1}{\eps})^{7/4}})$-time algorithm for the reduced problem, which is randomized and succeeds with high probability.
\end{lemma}

\bibliography{lipics-v2021-sample-article}

\clearpage
\appendix

\section{Details of Approximating Convolution}\label{app:conv}
We shall prove the following lemma.
\lemmerging*

We need Chi et al.'s algorithm for bounded monotone $(\min,+)$-convolution~\cite{CDXZ22}. We say a set $S$ of tuples is monotone if for any two tuple $(w, p)$, $(w', p')$ in $S$, $w' \geq w$ implies $p' \geq p$. We say $S$ is bounded by $O(\ell)$ if the entries of the tuples in $S$ are non-negative integers bounded by $O(\ell)$. 

\begin{lemma}\label{lem:bm-conv-accurate}{\normalfont \cite{CDXZ22}}
    Let $S_1$ and $S_2$ be two sets of tuples. Suppose that both $S_1$ and $S_2$ are monotone and bounded by $O(\ell)$. In $\tilde{O}(\ell^{3/2})$ time and with high probability\footnote{The original algorithm in~\cite{CDXZ22} runs in $\tilde{O}(\ell^{3/2})$ expected time. It is easy to convert it into an algorithm that succeeds with high probability. Such an algorithm is sufficient for our purpose since other parts of our algorithm only succeed with high probability.}, we can compute $S_1 \oplus S_2$.
\end{lemma}

With this algorithm, we can approximately compute the $(\max,+)$-convolution of two sets of tuples in $\tilde{O}((\frac{1}{\eps})^{3/2})$ time. The following lemma is essentially the same as that in~\cite[Lemma 28 in the full version]{BC22}, and its proof follows the idea of~\cite[Lemma 1]{Chan18}.

\begin{restatable}{lemma}{lembmconvapprox}
\label{lem:bm-conv-approx}
    Let $S_1, S_2 \subseteq (\{0\} \cup [A, B]) \times (\{0\} \cup [C, D])$ be two sets of tuples with total size $\ell$. In $\tilde{O}(\ell + (\frac{1}{\eps})^{3/2})$ time and with high probability, we can compute a set $S$ of size $\tilde{O}(\frac{1}{\eps})$ that approximates $S_1 \oplus S_2$ with factor $1 + \eps$.  Moreover, $S$ is dominated by $S_1 + S_2$.
\end{restatable}
    
\begin{proof}
    By scaling, we assume $A = C = 1$.  For a given $a, b \geq 1$, we first show how to compute a set of size $O(\frac{1}{\eps})$ that approximates $(S_1 \oplus S_2) \cap ([0,a] \times [0,b])$ with additive error $(\eps a, \eps b)$. 

    We assume that $S_1 \cup S_2 \subseteq [0,a]\times [0,b]$ since tuples outside this range can be safely removed.  For each tuple $(w,p) \in S_1 \cup S_2$, we round $w$ up to the nearest multiple of $\eps a$, and round $p$ down to the nearest multiple of $\eps b$. This step incurs an additive error of $(O(\eps a), O(\eps b))$. By scaling, we assume that all the tuple entries of $S_1$ and $S_2$ are integers bounded by $O(\frac{1}{\eps})$. All dominated tuples in $S_1$ and $S_2$ can be removed (this can be done in $O(\ell\log \ell)$ time by sorting), since they cannot contribute to $S_1 \oplus S_2$. After that, $S_1$ and $S_2$ become monotone and bounded by $O(\frac{1}{\eps})$. Then we compute $S_1 \oplus S_2$ by Lemma~\ref{lem:bm-conv-accurate}, which takes $\tilde{O}((\frac{1}{\eps})^{3/2})$ time and succeeds with high probability. It is also easy to see that the resulting set is dominated by the original $S_1 + S_2$ because rounding always makes a tuple worse (that is, with larger $w$ and less $p$).

    To finish the proof, we repeat the above procedure for every $a \in [1, \frac{B}{A}]$ and every $b \in [1, \frac{D}{C}]$ that are powers of $2$, and take the union of all resulting sets. (One can also remove the dominated tuples in the union in $\tilde{O}(\frac{1}{\eps})$ time, although this is not necessary due to Definition~\ref{def:approx-set}.) Note that for tuples $(w, p) \in [a/2, a]\times [b/2, b]$, the additive error $(O(\eps a), O(\eps b))$ implies an approximate factor of $1 + O(\eps)$. The running time increases by a factor of $O(\log \frac{B}{A} \cdot \log \frac{D}{C})$.
\end{proof}

Now we are ready to present an approximation algorithm for computing the $(\max,+)$-convolution of multiple sets. It follows the idea of~\cite[Lemma 2(i)]{Chan18}.
\begin{proof}[Proof of Lemma~\ref{lem:merging}]
    We use a divide-and-conquer approach:  compute $S_1 \oplus \cdots \oplus S_{m/2}$ and $S_{m/2+1} \oplus \cdots \oplus S_{m}$ recursively, and merge the two resulting sets by Lemma~\ref{lem:bm-conv-approx}. The recursion tree has $O(m)$ nodes, so the total running time is $\tilde{O}(\ell + (\frac{1}{\eps})^{3/2}m)$. ($\ell$ comes from the leaf nodes.) Since there are $O(\log m)$ levels, and each level of the tree incurs an approximation factor of $1 + O(\eps)$, the overall approximation factor is $1 + O(\eps \log m)$. We can adjust $\eps$ by a factor of $\log m$, and the running time increases by polylogarithmic factors.
\end{proof}

\section{Details of Approximation Errors}\label{app:err}
We prove the following two lemmas.
\lemapproxerra*
\begin{proof}
    We only give a proof for statement (i), and statement (ii) can be proved similarly. Let $(w,p)$ be an arbitrary tuple in $S$. By Definition~\ref{def:approx-set}, there exists $(w_1, p_1) \in S_1$ such that $w_1 \leq (1 + \eps_1)w $ and $p_1 \geq \frac{p}{1 + \eps_1}$, and there exists $(w_2, p_2) \in S_2$ such that $w_2 \leq (1 + \eps_2)w_1 $ and $p_2 \geq \frac{p_1}{1 + \eps_2}$. It is easy to see that $w_2 \leq (1 + \eps_1)(1 + \eps_2)w$ and that $p_2 \geq \frac{p}{(1 + \eps_1)(1 + \eps_2) }$. 
\end{proof}

\lemapproxerrb*
\begin{proof}
    We only give a proof for statement (i), and statement (ii) can be proved similarly. 

    We first prove that $S'_1 + S'_2$ approximates $S_1 + S_2$ with factor $1 + \eps$. Let $(w,p)$ be an arbitrary tuple in $S_1 + S_2$. By definition, $(w, p) = (w_1, p_1) + (w_2, p_2)$ for some $(w_1, p_1) \in S_1$ and $(w_2, p_2) \in S_2$.  By Definition~\ref{def:approx-set}, there exists $(w'_1, p'_1) \in S'_1$ such that $w'_1 \leq (1 + \eps)w_1 $ and $p'_1 \geq \frac{p_1}{1 + \eps}$, and there exists $(w'_2, p'_2) \in S'_2$ such that $w'_2 \leq (1 + \eps)w_2 $ and $p'_2 \geq \frac{p_2}{1 + \eps}$. Let $(w', p') = (w'_1, p'_1) + (w'_2, p'_2)$.  We have that $(w', p') \in S'_1 + S'_2$, $w' \leq (1 + \eps)w$, and $p' \geq p/(1 + \eps)$. 

    Next we prove that $S'_1 \oplus S'_2$ approximates $S_1 \oplus S_2$ with factor $1 + \eps$.  Let $(w,p)$ be an arbitrary tuple in $S_1 \oplus S_2$.  As in the above paragraph, there exists $(w', p') \in S'_1 + S'_2$ such that $w' \leq (1 + \eps)w$, and $p' \geq p/(1 + \eps)$.  If $(w', p') \in S'_1 \oplus S'_2$, then we are done. If not, it must be that $(w', p')$ is dominated by some tuple $(w'', p'') \in S'_1 \oplus S'_2$, but then $w'' \leq (1 + \eps)w$, and $p'' \geq p/(1 + \eps)$.
\end{proof}

\section{Details of Problem Reduction}\label{app:reducing}
Recall that the reduced problem is defined as follows.
\begin{quote}
The reduced problem $\mathrm{RP}(\alpha)$: given a set $I$ of $n$ items from $[1,2] \times [1,2]$ and a real number $\alpha \in [\frac{1}{2}, \frac{1}{4\eps}]$, compute a set $S$ that approximates ${\cal S}^+(I) \cap ([0,\frac{1}{\alpha\eps}] \times [0, \frac{1}{\alpha\eps}])$ with additive error $(\tilde{O}(\frac{1}{\alpha}), \tilde{O}(\frac{1}{\alpha}))$. Moreover, $S$ should be dominated by $\mathcal{S}(I)$.
\end{quote}

We shall prove the following lemma.
\lemreduce*

Let $(I, t)$ be a Knapsack instance. Let $\mathrm{OPT}$ be the optimal value. Recall that to obtain a weak approximation, it suffices to compute a set that approximates $\mathcal{S}^+(I) \cap ([0, t] \times [0, OPT])$ with additive error $(\eps t, \eps\cdot \mathrm{OPT})$. (Strictly speaking, the set should also be dominated by $\mathcal{S}(I)$. One can verify that this property follows quite straightforwardly. So we omit this property.) We can slightly relax the permissible additive error to $(\tilde{O}(\eps t), \tilde{O}(\eps\cdot \mathrm{OPT}))$ because an additive error of $(\tilde{O}(\eps t), \tilde{O}(\eps\cdot \mathrm{OPT}))$ can be reduced to $(\eps t, \eps\cdot \mathrm{OPT})$ by adjusting $\eps$ by a polylogarithmic factor and this increases the running time by only a logarithmic factor.

\subsection{Preprocessing}
We first preprocess the instance so that the following properties hold after the preprocessing.
\begin{itemize}
    \item $(w, p) \in [1, \frac{1}{\eps}] \times [1, \frac{1}{\eps}]$ for every item $(w, p) \in I$;

    \item $\frac{1}{2\eps} \leq \mathrm{OPT} \leq \frac{1}{\eps}$ and $t = \frac{1}{\eps}$.
\end{itemize}
The preprocessing takes $O(n \log n)$ time.

\subsubsection{Preprocessing Profits}
Knapsack has an $O(n)$-time greedy algorithm that achieves an approximation ratio of $2$~\cite[Section 2.5]{KPD04}. Let $y$ be the value of this greedy solution. Then $y \leq \mathrm{OPT} \leq 2y$. We say an item $(w, p)$ is cheap if $p \leq 2\eps y$. We shall merge cheap items into meta-items with large profits.  Let $I^c$ be the set of cheap items. We partition $I^c$ into $I^c_1, \ldots, I^c_\ell$ so that
\begin{romanenumerate}
    \item For every integer $j \in [1, \ell-1]$, the efficiency of an item in $I^c_j$ is no smaller than that of any item in $I^c_{j+1}$. That is, $\frac{p}{w} \geq \frac{p'}{w'}$ for any $(w, p) \in I^c_j$ and any $(w', p') \in I^c_{j+1}$.

    \item $2\eps y < p(I^c_j) \leq 4\eps y$ for every integer $j \in [1, \ell - 1]$ and $p(I^c_\ell) \leq 4\eps y$. 
\end{romanenumerate}
The partition can be done in $O(n \log n)$ time by scanning the items in $I^c$ in decreasing order of efficiency. The items in $I^c_\ell$ can be safely discarded. This decreases the optimal value by at most $4\eps y$.  For every remaining group, $2\eps y < p(I^c_j) \leq 4\eps y$, and we replace it with a single meta-item $(p(I^c_j), w(I^c_j))$. We claim that replacing all $I^c_j$'s with meta-items decreases the optimal value by at most $4\eps y$. To see this, consider any subset $S$ of cheap items. By selecting meta-items in decreasing order of efficiency, we can always obtain a subset $S'$ of meta-items such that $p(S') \geq p(S) - 4\eps y$ and $w(S') \leq w(S)$.

Now every item has $p_i > 2\eps y$. Without loss of generality, we assume that $p_i \leq 2y$ (since $OPT \leq 2y$). Then we scale all the profit by $2\eps y$. That is, we set $p_i:=  \frac{p_i}{2\eps y} $ for every $p_i$.  After scaling, every item has $p_i \in (1, \frac{1}{\eps}]$, and $\frac{1}{2\eps} \leq \mathrm{OPT} \leq \frac{1}{\eps}$.

\subsubsection{Preprocessing Weights}
Preprocessing weights is similar to preprocessing profits. We say an item is small if $w_i \leq \eps t$. Let $I^s$ be the set of small items.  We partition $I^s$ into $I^s_1, \ldots, I^s_\ell$ so that
\begin{romanenumerate}
    \item For every integer $j \in [1, \ell-1]$, the efficiency of an item in $I^s_j$ is no smaller than that of any item in $I^s_{j+1}$. That is, $\frac{p}{w} \geq \frac{p'}{w'}$ for any $(w, p) \in I^s_j$ and any $(w', p') \in I^s_{j+1}$.

    \item $\eps t < p(I^s_j) \leq 2\eps t$ for every integer $j \in [1, \ell - 1]$ and $p(I^s_\ell) \leq 2\eps t$.
\end{romanenumerate}
It can be done in $O(n \log n)$ time by scanning the items in $I^s$ in decreasing order of efficiency.  The items in $I^s_\ell$ can be ignored, since we can always select these items using an additional capacity of at most $2\eps t$. For every remaining group, $\eps t < p(I^s_j) \leq 2\eps t$, and we replace it with a single meta-item $(p(I^s_j), w(I^s_j))$. We claim that replacing all $I^s_j$'s with meta-items increases the given capacity by at most $2\eps t$. To see this, consider any subset $S$ of small items. By selecting meta-items in decreasing order of efficiency, we can always obtain a subset $S'$ of meta-items such that $p(S') \geq p(S)$ and $w(S') \leq w(S) + 2\eps t$.

Now every item has $w_i > \eps t$. Without loss of generality, we assume that $w_i \leq t$. Then we scale all the weights by $\eps t$. That is, we set $w_i:=  \frac{w_i}{\eps t} $ for every $w_i$.  After scaling, every item has $w_i \in (1, \frac{1}{\eps}]$, and $t = \Theta(\frac{1}{\eps})$.

\subsection{Partitioning Items into Groups}
After preprocessing, $\frac{1}{2\eps} \leq OPT \leq \frac{1}{\eps}$ and $t = \frac{1}{\eps}$. Our goal becomes computing a set that approximates ${\cal S}^+(I) \cap ([0, \frac{1}{\eps}]\times [0, \frac{1}{\eps}])$ with additive error $(\tilde{O}(1), \tilde{O}(1))$.  We shall compute such a set by the following three steps.
\begin{romanenumerate} 
    \item Recall that, after preprocessing, all the tuples in $I$ are in $[1, \frac{1}{\eps}] \times [1, \frac{1}{\eps}]$. Partition $I$ into $O(\log^2 \frac{1}{\eps})$ groups so that within each group, all item weights differ by a factor of at most $2$, and all item profits differ by a factor of at most $2$.

    \item for each group $I'$, compute a set $S'$ that approximates ${\cal S}^+(I') \cap ([0, \frac{1}{\eps}]\times [0, \frac{1}{\eps}])$ with additive error $(\tilde{O}(1), \tilde{O}(1))$. 

    \item merging the resulting sets via Lemma~\ref{lem:merging}.
\end{romanenumerate}
Step (i) costs $\tilde{O}(n)$ time. Since there are only $O(\log^2 \frac{1}{\eps})$ groups, Step (iii) costs $\tilde{O}((\frac{1}{\eps})^{3/2})$ time and increases the additive error by polylogarithmic factors.

It remains to consider Step (ii). Let $I'$ be an arbitrary group.  The tuples in $I'$ are from $[a, 2a]\times [b, 2b]$ for some $a,b \in [1, \frac{1}{2\eps}]$. Let $\alpha = \frac{\max\{a, b\}}{2}$. Note that $\alpha \in [\frac{1}{2},\frac{1}{4\eps}]$. By scaling, we can assume that every $(w, p) \in [1, 2] \times [1,2]$. After the scaling, the permissible additive error becomes $(\tilde{O}(1/a), \tilde{O}(1/b)) \geq (\tilde{O}(1/\alpha), \tilde{O}(1/\alpha))$. The part of $S^+(I')$ that needs to be approximated becomes $S^+(I') \cap ([0, \frac{1}{a\eps}]\times [0, \frac{1}{b\eps}])$.  Since the tuples in $I'$ are from $[1, 2] \times [1,2]$, we have $\frac{1}{2} \leq \frac{p}{w} \leq 2$ for every $(w, p) \in {\cal S}^+(I')$.  Therefore, 
\[
    (S^+(I') \cap ([0, \frac{1}{a\eps}]\times [0, \frac{1}{b\eps}])) \subseteq 
    S^+(I') \cap ([0, \frac{1}{\alpha\eps}]\times [0, \frac{1}{\alpha\eps}]).
\]
In conclusion, all we need to do is to approximate $S^+(I') \cap ([0, \frac{1}{\alpha\eps}]\times [0, \frac{1}{\alpha\eps}])$ with additive error $(\tilde{O}(\frac{1}{\alpha}), \tilde{O}(\frac{1}{\alpha}))$.

\section{Details of the Algorithm for Bounded Solution Size}\label{app:color-coding}
We shall prove the following lemma. 

\lemalgforlargealpha*


Since the items in $I$ are from $[1,2]\times [1,2]$, we have that  $|X| \leq \frac{1}{\alpha\eps}$ for any $X \subseteq I$ with $(w(X), p(X)) \in {\cal S}^+(I;\frac{1}{\alpha\eps})$. To deal with this case where the solution size is small, we use color coding. To reduce the running time, Bringmann~\cite{Bri17} proposed a two-layer color-coding technique, leading to an $\tilde{O}(n + t)$-time algorithm for Subset Sum. This technique was also used in later works for Subset Sum~\cite{CLMZ24cSTOCPartition,CLMZ24FOCS}.

Let $X$ be a subset of $I$ with $(w(X), p(X)) \in {\cal S}^+(I;\frac{1}{\alpha\eps})$.  Roughly speaking, if we randomly partition $I$ into $\tilde{O}(\frac{1}{\alpha\eps})$ subsets, then with high probability, each subset contains at most one element from $X$. After adding $(0,0)$ to each subset, and merging them via Lemma~\ref{lem:merging}, we can obtain a set that (approximately) contains $(w(X), p(X))$ in $\tilde{O}((\frac{1}{\eps})^{5/2}\frac{1}{\alpha})$ time.

The standard color-coding is based on balls-and-bins analysis, and can be summarized by the following fact.
\begin{fact}\label{fact:colorcoding}
    Let $I$ be a set of items. Let $k>0$ be an integer. Let $I_1 \cup \cdots \cup I_{k^2}$ be a random partition of $I$. For any $X\subseteq I$ and $|X|\leq k$, with probability at least $1/4$, we have that
    \(
        | X \cap I_i| \leq 1
    \)
    for all $i \in \{1, \ldots, k^2\}$.
\end{fact}

\begin{lemma}\label{lem:secend-layer-color-coding-algorithm}
    Let $I$ be a set of items. Let $k>0$ be an integer. In $\tilde{O}(n + (\frac{1}{\eps})^{3/2}k^2)$ time, we can compute a set $\tilde{S}$ of size $\tilde{O}(\frac{1}{\eps})$ that for any $X \subseteq I$ with $|X| \leq k$ and $(w(X),p(X))\in \mathcal{S}^+(I)$, with high probability, $\tilde{S}$ approximates $(w(X), p(X))$ with factor $1 + O(\eps)$.
\end{lemma}
\begin{algorithm}
\caption{$\mathtt{ColorCoding}(I, k, q)$~\cite{Bri17}}
\label{alg:color-coding}
    \begin{algorithmic}[1]
    \Statex \textbf{Input:} A set of $n$ tuples $I$, a positive integer $k$, and a error probability $\delta$
    \Statex \textbf{Output:} a set of $S$
        \For{$j =1,\ldots,\lceil \log \frac{1}{q} \rceil$}
           \State Randomly partition $I$ into $k^2$ subsets $I^j_1, \ldots, I^j_{k^2}$\;
           \State Compute $\tilde{S}_j$ that approximates $(I^j_1\cup\{(0,0)\}) \oplus \cdots \oplus (I^j_{k^2}\cup\{(0,0)\})$ with factor $1+O(\eps)$ by Lemma~\ref{lem:merging}.
        \EndFor
        \State Compute $\tilde{S} = \cup_j \tilde{S}_j$
  \State \Return $\tilde{S}$
  \end{algorithmic}
\end{algorithm}
\begin{proof}
    
    We fix $q: = (n + \frac{1}{\eps})^{-\Omega(1)}$ and let $\tilde{S}$ be the output of $\mathtt{ColorCoding}(I, k, q)$ in Algorithm~\ref{alg:color-coding}. In this proof, we assume that line 3 of  Algorithm~\ref{alg:color-coding} always succeeds, and this happens with high probability.
    
    We first show that the running time of Algorithm~\ref{alg:color-coding} is  $\tilde{O}(n + (\frac{1}{\eps})^{3/2}k^2)$ time. Each iteratin costs $\tilde{O}(n + (\frac{1}{\eps})^{3/2}k^2)$ time and there are $O(\log \frac{1}{q})$ iterations, so all the iterations costs  $\tilde{O}(n + (\frac{1}{\eps})^{3/2}k^2)$ time. Line 4 costs $\tilde{O}(\frac{1}{\eps})$ time. Also, each $S_j$ is of size $\tilde{O}(\frac{1}{\eps})$ by Lemma~\ref{lem:merging} and there are $O(\log \frac{1}{q})$ iterations, so the size of $S$ is $\tilde{O}(\frac{1}{\eps})$.

    For any $X\subseteq I$ with $|X|\leq k$,
    by Fact~\ref{fact:colorcoding}, for each $j\in\{1,\ldots,\lceil \log \frac{1}{q} \rceil\}$, with probability at least $1/4$, $|X\cap I^j_i|\leq 1$ for all $i\in \{1, \ldots, k^2\}$. That is, $(w(X\cap I^j_i),p(X\cap I^j_i))\in (I^j_i\cup\{(0,0)\})$. Since $(w(X),p(X))\in \mathcal{S}^+(I)$, it is easy to see that with probability at least $1/4$, $(w(X),p(X))\in (I^j_1\cup\{(0,0)\}) \oplus \cdots \oplus (I^j_{k^2}\cup\{(0,0)\})$. Hence, with probability at least $1/4$, $\tilde{S_j}$ approximates $(w(X), p(X))$ with factor $1 + O(\eps)$.

   Since $\tilde{S}=\cup_{j}\tilde{S}_j$, it suffices that $(w(X), p(X))$ is approximated by some $\tilde{S_j}$. It is easy to see that with probability at least $1-(3/4)^{\lceil \log \frac{1}{q} \rceil} \geq 1-q$, there exists some $j^*\in \{1,\ldots,\lceil \log \frac{1}{q} \rceil\}$ such that $\tilde{S}_{j^*}$ approximates $(w(X), p(X))$ with factor $1 + O(\eps)$. So with high probability, $\tilde{S}$ approximates $(w(X), p(X))$ with factor $1 + O(\eps)$.
\end{proof}

However, the running time of Algorithm~\ref{alg:color-coding} depends on $k^2$, which is too large for our purpose. So we use the two-layer color-coding technique proposed by Bringmann~\cite{Bri17}, which reduces the running time by a factor of $k$.

\begin{lemma}[\cite{Bri17}]\label{lem:first-layer-color-coding}
    Let $I$ be a set of items. Let $k>0$ be an integer, and let $q\in (0,1)$. Let $I_1,\ldots,I_{k}$ be a random partition of $I$. For any $X\subseteq I$ with $|X|\leq k$, with probability at least $1-q$, we have that
    \(
        |X\cap I_j| \leq 6\log\frac{k}{q}
    \)
    for all $j \in \{1,\ldots, k\}$.
\end{lemma}

\begin{lemma}\label{lem:first-layer-color-coding-algorithm}
    Let $I$ be a set of items. Let $k>0$ be an integer. In $\tilde{O}(n + (\frac{1}{\eps})^{3/2}k)$ time, we can compute a set $\tilde{S}$ of size $\tilde{O}(\frac{1}{\eps})$ that for any $X \subseteq I$ with $|X| \leq k$ and $(w(X),p(X))\in \mathcal{S}^+(I)$, with high probability, $\tilde{S}$ approximates $(w(X), p(X))$ with factor $1 + O(\eps)$.
\end{lemma}
\begin{proof}
    We fix $q: = (n + \frac{1}{\eps})^{-\Omega(1)}$ and randomly partition $I$ into $k$ subsets $I_1,\ldots,I_{k}$. 
    
    For each $I_i$, by Lemma~\ref{lem:secend-layer-color-coding-algorithm}, we can compute a 
    set $\tilde{S}_i$ of size $\tilde{O}(\frac{1}{\eps})$ in $\tilde{O}(|I_i|+(\frac{1}{\eps})^{3/2}(6\log\frac{k}{q})^2)$ time such that for any $X' \subseteq I_i$ with $|X'| \leq 6\log\frac{k}{q}$ and $(w(X'),p(X'))\in \mathcal{S}^+(I_i)$, with high probability, $\tilde{S}_i$ approximates $(w(X'), p(X'))$ with factor $1 + O(\eps)$.
    
    For any $X \subseteq I$ with $|X| \leq k$ and $(w(X),p(X))\in \mathcal{S}^+(I)$, by Lemma~\ref{lem:first-layer-color-coding}, with probability at least $1-q$, $|X\cap I_i|\leq 6\log\frac{k}{q}$ and $(w(X\cap I_i),p(X\cap I_i))\in \mathcal{S}^+(I_i)$. Then with high probability, $\tilde{S}_i$ approximates $(w(X\cap I_i), p(X\cap I_i))$ with factor $1 + O(\eps)$.

    By Lemma~\ref{lem:merging}, in $\tilde{O}(n+(\frac{1}{\eps})^{3/2}k)$ time and with high probability, we can compute a set $\tilde{S}$ of size $\tilde{O}(\frac{1}{\eps})$ that approximates $\tilde{S}_1\oplus\cdots\oplus\tilde{S}_k$ with factor $1 + O(\eps)$. It is easy to see that with high probability, $\tilde{S}$ approximates $(w(X),p(X))$ with factor $1 + O(\eps)$.

    The overall time is $\tilde{O}(n+(\frac{1}{\eps})^{3/2}k)$.
\end{proof}

Since for any $(w(X),p(X))\in {\cal S}^+(I;\frac{1}{\alpha\eps})$, $|X|\leq \frac{1}{\alpha\eps}$, we can prove Lemma~\ref{lem:alg-for-large-alpha}.

\begin{proof}[Proof of Lemma~\ref{lem:alg-for-large-alpha}]
    For any $(w(X),p(X))\in {\cal S}^+(I;\frac{1}{\alpha\eps})$, we have $|X|\leq \frac{1}{\alpha\eps}$. So by Lemma~\ref{lem:first-layer-color-coding-algorithm},
    in $\tilde{O}(n + (\frac{1}{\eps})^{5/2}\frac{1}{\alpha})$ time, we can compute a set of size $\tilde{O}(\frac{1}{\eps})$ that for any $(w(X),p(X))\in {\cal S}^+(I;\frac{1}{\alpha\eps})$, with high probability, approximates $(w(X),p(X))$ with factor $1 + O(\eps)$.
    
    To approximate every $(w(X),p(X))\in {\cal S}^+(I;\frac{1}{\alpha\eps})$ with high probability, we can round all weights and profits in $I$ to the multiple of $\eps$ with a factor $1+\eps$. Since for any $(w,p)\in\mathcal{S}^+(I;\frac{1}{\alpha\eps})$, $w\leq 2n$ and $\mathcal{S}^+(I;\frac{1}{\alpha\eps})$ is monotone, we have $|\mathcal{S}^+(I;\frac{1}{\alpha\eps})|\leq \frac{2n}{\eps}$.
    Then by union bound it follows directly that with high probability, $\tilde{S}$ approximates ${\cal S}^+(I; \frac{1}{\alpha\eps})$ with factor $1 + O(\eps)$. 
\end{proof}

\section{Details of the Proximity Bound}\label{app:proximity}
We shall prove the following lemma.
\lemproximity*
    
\begin{proof}
    We label the items in $I$ as  $\{(w_1, p_1), \ldots, (w_n, p_n)\}$ in decreasing order of efficiency. (Recall that the labeling is unique by fixing a total order on the items to break ties.)  Let $(w_b, p_b)$ be the breaking item with respect to $w^*$. If $b = n + 1$, then $I^* = I$ since $w_1 + \cdots + w_n \leq w^*$. The lemma holds straightforwardly. Assume that $b \leq n$.  Let $B = \{(w_1, p_1), \ldots, (w_{b-1}, p_{b-1})\}$. By the definition of breaking item, we have $w(I^*) < w(B)+w_{b}$. It implies
    \begin{equation}\label{eq:weight-bound}
        w(I^*\setminus B)   <   w(B\setminus I^*) + w_b.
    \end{equation}
    Also, it must be that $p(B) \leq p^*=p(I^*)$ since otherwise the tuple $(w^*, p^*)$ would be dominated by $(w(B),p(B))$ and should not be in ${\cal S}^+(I)$. This implies
    \begin{equation}\label{eq:profit-bound}
         p(B\setminus I^*) \leq p(I^*\setminus B).
    \end{equation}
    \begin{romanenumerate}
        \item It is easy to see that $I'\subseteq I\setminus B$. Let $X=I\setminus (B\cup I')$.
        We have 
        \begin{equation}
           p(I^*\setminus B) = p(I^* \cap (I\setminus B)) = p(I^*\cap I') + p(I^*\cap X) \label{eq:profit-1}.
        \end{equation}
        Note that the items in $I'$ have efficiency at most $\rho_b - \Delta$, the items in $B$ have efficiency at least $\rho_b$, and the items in $X$ have efficiency at most $\rho_b$.  We have
        \begin{align*}
            w(I^*\cap I')\cdot (\rho_b-\Delta)&\geq p(I^*\cap I') \qquad\qquad\qquad &\text{By Eq~\eqref{eq:profit-1}}\\
            &= p(I^*\setminus B)-p(I^*\cap X)  \qquad\qquad\qquad &\text{By Eq~\eqref{eq:profit-bound}}\\
            &\geq p(B\setminus I^*)-p(I^*\cap X)    \\
            &\geq w(B\setminus I^*)\cdot \rho_b - w(I^*\cap X)\cdot \rho_b &\qquad\text{By Eq~\eqref{eq:weight-bound}}\\
            &> (w(I^*\setminus B)-w_b)\cdot \rho_b - w(I^*\cap X)\cdot \rho_b  \\
            &= w(I^*\cap I')\cdot \rho_b-p_b.
        \end{align*}        
        So we have $w(I^*\cap I') < \frac{p_b}{\Delta}\leq \frac{2}{\Delta}$.
        \item 
        Clearly, $I' \subseteq B$. Let $X = B \setminus I'$. We have 
        \begin{equation}
            p(B\setminus I^*) = p(I'\setminus I^*) + p(X\setminus I^*)\label{eq:profit-2}.
        \end{equation}
        Note that the items in $I'$ have efficiency at least $\rho_b + \Delta$, the items in $I\setminus B$ have efficiency at most $\rho_b$, and the items in $X$ have efficiency at least $\rho_b$.  We have
        \begin{align*}
            w(I'\setminus I^*)\cdot (\rho_b+\Delta) &\leq  p(I'\setminus I^*)\qquad\qquad\qquad&\text{By Eq~\eqref{eq:profit-2}}\\
            &=p(B\setminus I^*)-p(X\setminus I^*)\qquad\qquad\qquad &\text{By Eq~\eqref{eq:profit-bound}}\\
            &\leq p(I^*\setminus B)-p(X\setminus I^*) \\
            &\leq w(I^*\setminus B)\cdot\rho_b-w(X\setminus I^*)\cdot \rho_b \quad &\text{By Eq~\eqref{eq:weight-bound}}\\
            &\leq (w(B\setminus I^*)+w_b)\cdot\rho_b-w(X\setminus I^*)\cdot \rho_b  \\
            & =w(I'\setminus I^*)\cdot \rho_b+p_b
        \end{align*}
        So we have $w(I'\setminus I^*) < \frac{p_b}{\Delta} \leq \frac{2}{\Delta}$. \qedhere
    \end{romanenumerate}
\end{proof}

\section{An Improved Algorithm with Exponent 7/4 (Full Version)}\label{app:strong}
In our $\tilde{O}(n + (\frac{1}{\eps})^{11/6})$-time algorithm, only the bad groups benefit from the proximity bound, which allows us to approximate with larger factor $1 + \frac{1}{\alpha\tau}$, while for the good groups, we simply approximate with factor $1 + \eps$. To further improve the running time, we shall partition the items in a way that all the groups can benefit from the proximity bound.

We assume that $\frac{1}{2}\leq \alpha \leq (\frac{1}{\eps})^{3/4}$ since the algorithm in Lemma~\ref{lem:alg-for-large-alpha} already runs in $\tilde{O}(n + (\frac{1}{\eps})^{7/4})$ time for all $\alpha \geq (\frac{1}{\eps})^{3/4}$. We assume that $I = \{(w_1, p_1), \ldots, (w_n, p_n)\}$ and that $\frac{p_1}{w_1} \geq \cdots \geq \frac{p_n}{w_n}$.

We first partition $I$ into the head group $I_{\mathrm{head}}$ and the tail group $I_{\mathrm{tail}}$ where $I_{\mathrm{head}}$ contains the first $\lceil \frac{1}{\alpha\eps}\rceil + 1 $ items and $I_{\mathrm{tail}}$ contains the rest of the items.

\subsection{Approximating Head Group}\label{sec:strong-head}
Let $n' = \lceil \frac{1}{\alpha\eps}\rceil + 1$.  We shall show that in $\tilde{O}((\frac{1}{\eps})^{7/4})$ time, we can compute a set $S$ of size $\tilde{O}(\frac{1}{\eps})$ that approximate $\mathcal{S}^+(I_{\mathrm{head}})$ with additive error $O(\frac{1}{\alpha})$.

Let $\tau$ be a parameter that will be specified later. Roughly speaking, we will further partition $I_{\mathrm{head}}$ into good groups and bad groups. The bad groups are the same as before: a bad group is a group of $\tau$ items whose efficiencies differ by at least $\frac{1}{\tau}$. For good groups, we will make them larger than before: when we obtain a good group $I'$ of size $\tau$, we will keep adding items to $I'$ until the difference between the maximum and the minimum item efficiencies exceeds $\frac{1}{|I'|}$ with the next item. A more precise description of the partition process is given below. 

We partition $I_{\mathrm{head}}$ into $I_1, \ldots, I_m$ as follows. Initially, $j = 1$ and $k = 1$, and the remaining items are $\{(w_k, p_k), \ldots, (w_{n'}, p_{n'})\}$. Let $k'$ be the minimum integer such that
\[
    (k' - k) \cdot (\frac{p_k}{w_k} - \frac{p_{k'}}{w_{k'}}) > 1.
\]
If such $k'$ does not exists, we set $I_j = \{(w_k, p_k), \ldots, (w_{n'}, p_{n'})\}$, set $\Delta_j = \Delta'_j = \frac{p_k}{w_k} - \frac{p_{n'}}{w_{n'}}$, set $m: = j$, and finish. Assume that $k'$ exists. 
\begin{itemize}
    \item If $k' - k \geq \tau$, we set $I_j = \{(w_k, p_k), \ldots, (w_{k' - 1}, p_{k'-1})\}$, set $\Delta_j = \frac{p_k}{w_k} - \frac{p_{k'-1}}{w_{k'-1}}$, set $\Delta'_j = \frac{p_k}{w_k} - \frac{p_{k'}}{w_{k'}}$, and proceed with $j := j + 1$ and $k: = k'$. In this case, we say $I_j$ is a good group. 

    \item Otherwise, we let $I_j = \{(w_k, p_k), \ldots, (w_{k + \tau - 1}, p_{k + \tau - 1})\}$, set $\Delta_j = \frac{p_k}{w_k} - \frac{p_{k + \tau - 1}}{w_{k+\tau - 1}}$, set $\Delta'_j = \frac{p_k}{w_k} - \frac{p_{k + \tau}}{w_{k+\tau}}$, and proceed with $j: = j + 1$ and $k:= k + \tau$. In this case, we say that $I_j$ is a bad group. 
\end{itemize}
(We remark that $\Delta_j$ and $\Delta'_j$ are not required by the algorithm. They are maintained only for the purpose of analysis. $\Delta_j$ is the actual difference between the maximum and minimum item efficiencies in $I_j$. For technical reasons, we also need $\Delta'_j$. Basically, in addition to $\Delta_j$, the gap $\Delta'_j$ also includes the efficiency gap between the minimum efficiency in $I_j$ and the maximum efficiency in $I_{j+1}$. We will use $\Delta_j$ to bound the time cost of $I_j$, and use $\Delta'_j$ to create an efficiency gap for other groups.)

We have the following observations.
\begin{observation}\label{obs:group-property}
The groups $I_1, \ldots, I_m$ satisfy the following properties.
\begin{romanenumerate}
    \item $m \leq \frac{1}{\alpha\eps\tau} + 1$

    \item $\sum_{j=1}^{m}|I_j| = |I_{\mathrm{head}}| = \frac{1}{\alpha\eps} + 1$ and $\sum_{j=1}^{m}\Delta_j \leq \sum_{j=1}^{m}\Delta'_j \leq \frac{3}{2}$.

    \item For every group $I_j$, the efficiencies of the items in $I_j$ differ by at most $\Delta_j$.

    \item For every good group $I_j$, we have $|I_j|\Delta_j \leq 1$ and $|I_j|\Delta'_j \geq \frac{1}{2}$.

    \item For every bad group $I_j$, we have $|I_j|\Delta'_j \geq |I_j|\Delta_j > 1$.

    \item For the last group $I_m$, we have $|I_m|\Delta_m \leq 1$.
\end{romanenumerate}
\end{observation}

For each group, we shall approximate it in exactly the same way as we did for the bad groups in the $\tilde{O}(n + (\frac{1}{\eps})^{11/6})$-time algorithm, except that we shall use $\frac{1}{\alpha|I_j|}$ as the accuracy parameter. More specifically, we use Corollary~\ref{coro:approx-add} and Corollary~\ref{coro:approx-del} to compute a set $S_j$ of size $\tilde{O}(\alpha|I_j|)$ that approximate each tuple $(w, p) \in \mathcal{S}^+(I_j)$ with additive error $\frac{1}{\alpha|I_j|} \cdot \min\{w, w(I) - w\}$. The time cost for $I_j$ is $\tilde{O}(|I_j| + \alpha|I_j| + \alpha^2|I_j|^2\Delta_j)$.

Then we compute a set $S_{\mathrm{head}}$ of size $\tilde{O}(\frac{1}{\eps})$ that approximates $S_1 \oplus \cdots \oplus S_m$ with factor $1+O(\eps)$ via Lemma~\ref{lem:merging}. 

\paragraph*{Bounding Time Cost}
Let $J_{\mathrm{good}}$ be the set of indices of good groups, and let $J_{\mathrm{bad}}$ be the set of indices of bad groups.
 
For the good groups and the last group, the total time cost is
\[
    \sum_{j \in J_{\mathrm{good}} \cup \{m\}} \tilde{O}(|I_j| + \alpha|I_j| + \alpha^2|I_j|^2\Delta_j) \leq \sum_{j \in J_{\mathrm{good}}\cup\{m\}}\tilde{O}( \alpha^2|I_j|)  \leq \tilde{O}( \frac{\alpha}{\eps}).
\]
The first inequality is due to that $|I_j|\Delta_j \leq 1$ (Observation~\ref{obs:group-property}(iv) and (vi)), and the last inequality is due to that $\sum_{j}|I_j| = \frac{1}{\alpha\eps} + 1$ (Observation~\ref{obs:group-property}(ii)).

For the bad groups, the total time cost is
\[
    \sum_{j \in J_{\mathrm{bad}}} \tilde{O}(|I_j| + \alpha|I_j| + \alpha^2|I_j|^2\Delta_j) \leq \sum_{j \in J_{\mathrm{bad}}} \tilde{O}(\alpha\tau + \alpha^2\tau^2\Delta_j) \leq \tilde{O}(m\alpha\tau + \alpha^2\tau^2).
\] 
The first inequality is due to that every bad group has size exactly $\tau$, and the second inequality is due to that $\sum_{j}\Delta_j \leq 3/2$ (Observation~\ref{obs:group-property}(ii)).

For the merging step, note that each set $S_j$ is of size $\tilde{O}(\alpha|I_j|)$ except for that $S_m$ is of size $\tilde{O}(\frac{1}{\eps})$. (Recall that $\sum_{j}|I_j| = \frac{1}{\alpha\eps} + 1$.) According to Lemma~\ref{lem:merging}, the time cost  is
\[
    \tilde{O}(\frac{1}{\eps} + \alpha\sum_j|I_j| + m\alpha\tau + (\frac{1}{\eps})^{3/2}m) \leq \tilde{O}(\frac{1}{\eps} + m\alpha\tau + (\frac{1}{\eps})^{3/2}m).
\]

So the total time cost to approximate $\mathcal{S}^+(I_{\mathrm{head}})$ is 
\[
    \tilde{O}(\frac{\alpha}{\eps}) + \tilde{O}(m\alpha\tau + \alpha^2\tau^2) +  \tilde{O}(\frac{1}{\eps} + m\alpha\tau + (\frac{1}{\eps})^{3/2}m)\leq \tilde{O}(\frac{\alpha}{\eps}+m\alpha\tau+\alpha^2\tau^2+(\frac{1}{\eps})^{3/2}m).
\]
Recall our assumption that $\frac{1}{2}\leq \alpha \leq (\frac{1}{\eps})^{3/4}$ and Observation~\ref{obs:group-property}(i) that $m \leq \frac{1}{\alpha\tau\eps} + 1$. By setting $\tau = \lceil (\frac{1}{\eps})^{5/6}\frac{1}{\alpha}\rceil$, one can verify that all the above time costs are bounded by $\tilde{O}((\frac{1}{\eps})^{7/4})$.

\paragraph*{Bound Additive Error}
We shall show that $S_{\mathrm{head}}$ approximates every tuple $(w^*, p^*) \in \mathcal{S}^+(I_{\mathrm{head}})$ with additive error $\tilde{O}(\frac{1}{\alpha})$. 

Let $(w^*, p^*)$ be an arbitrary tuple in $\mathcal{S}^+(I_{\mathrm{head}})$. Let $(w_b, p_b)$ be the breaking item with respect to $w^*$. Let $\rho_b = \frac{p_b}{w_b}$ be the efficiency of the breaking item. 

To apply the proximity bound, let $I_{j'}$ be the group containing the breaking item $(w_b, p_b)$. To apply the proximity bound, we will show that the groups can be divided into  $O(\log \frac{1}{\eps})$ collections so that for each collection $\{I_{j_{k} + 1}, I_{j_{k} + 2}, \ldots, I_{j_{k+1}-1} \}$,  the efficiency gap between $\rho_b$ and the groups in this collection is at least the maximum $\Delta'_j$ of these groups. To do this, we need the following auxiliary lemma.

\begin{lemma}\label{lem:cluster}
    Let $\Delta_1, \ldots, \Delta_n$ be a sequence of positive real numbers. Let $\Delta_{\min}$ and $\Delta_{\max}$ be the minimum and maximum numbers in the sequence. There exists $h = O(\log \frac{\Delta_{\max}}{\Delta_{\min}})$ indices $1 = j_1 < j_2 < \cdots < j_h = n$ such that for any $k \in \{1, \ldots, h-1\}$, we have that
    \begin{equation}\label{eq:cluster-original}
            \max\{\Delta_j : j_k < j < j_{k+1}\} \leq \sum\{\Delta_j : j_{k+1} \leq j \leq n\},
    \end{equation}
    where the maximum of an empty set is defined to be $0$.
\end{lemma}
\begin{proof}
    We start with $n$ indices that $j_k=k$ for all $k$ and iteratively remove indices from the set while preserving Eq~\eqref{eq:cluster-original}. 
   
   At first, we have $h=n$. For any $k \in \{1, \ldots, h-2\}$, if moving $\Delta_{j_{k+1}}$ from the right side of Eq~\eqref{eq:cluster} to its left side does not break the inequality, that is,
    \[
            \max\{\Delta_j : j_k < j \leq j_{k+1}\} \leq \sum\{\Delta_j : j_{k+1} < j \leq n\},
    \]
    then we update $j_{k+1}$ as $j_{k+1}\leftarrow j_{k+1}+1$, and note that Eq~\eqref{eq:cluster-original} is always maintained. Moreover, after updating if it happens that $j_{k+1}$ is the same as some other $j_{k'}$, then we simply delete $j_{k+1}$.  

  We repeat the above procedure until no updating is possible, in which case we have
    \begin{equation}\label{eq:max-geq-sum}
            \max\{\Delta_j : j_k < j \leq j_{k+1}\} > \sum\{\Delta_j : j_{k+1} < j \leq n\},
    \end{equation}
    for any $k \in \{1, \ldots, h-2\}$. We show that $h = O(\log\frac{\Delta_{\max}}{\Delta_{\min}})$ in this case. Note that 
    \begin{align*}
        &\sum\{\Delta_j : j_{k} < j \leq n\} \\
        =&\sum\{\Delta_j : j_{k} < j \leq j_{k+1}\}+\sum\{\Delta_j : j_{k+1} < j \leq n\}\\
        \geq& \max\{\Delta_j : j_k < j \leq j_{k+1}\} + \sum\{\Delta_j : j_{k+1} < j \leq n\} 
        \\>& 2 \sum\{\Delta_j : j_{k+1} < j \leq n\}, 
    \end{align*}
    for any $k \in \{1, \ldots, h-2\}$,  where the last inequality is due to Eq~\eqref{eq:max-geq-sum}. Then we have 
    \[
        \max\{\Delta_j : j_1 < j \leq j_{2}\} > \sum\{\Delta_j : j_{2} < j \leq n\} > 2^{h-3}\sum\{\Delta_j : j_{h-1} <  j \leq n\} \geq 2^{h-3}\Delta_n.
    \]
    Therefore, $h \leq \log_2 \frac{\Delta_{\max}}{\Delta_{\min}} + 3$.
\end{proof}

\begin{lemma}
    $S_1 \oplus \cdots \oplus S_m$ approximates $(w^*, p^*)$ with additive error $\tilde{O}(\frac{1}{\alpha})$.
\end{lemma}
\begin{proof}
    Since $I_{\mathrm{head}} = I_1 \cup \cdots \cup I_m$ and $(w^*, p^*) \in \mathcal{S}^+(I_{\mathrm{head}})$, we have that there is a tuple $(w^*_j, p^*_j) \in \mathcal{S}^+(I_j)$ for each $j$ such that 
    \[
        (w^*, p^*) = (w^*_1, p^*_1) + \cdots + (w^*_m, p^*_m).
    \]
    Recall that $S_j$ approximates $(w^*_j, p^*_j)$ with the additive error $\delta_j = \frac{1}{\alpha|I_j|}\cdot\min\{w^*_j,w(I_j)-w^*_j\}$. It suffices to bound $\sum_j \delta_j$. Let $I_{j'}$ be the group containing the breaking item $(w_b, p_b)$. Since the item weights are in $[1,2]$, it is easy to see that for any $j$
    \[
        \delta_{j} \leq \frac{1}{\alpha|I_{j}|}w^*_{j} \leq \frac{1}{\alpha}.
    \]
    Therefore, $\delta_{j'} + \delta_{m} \leq \frac{2}{\alpha}$. It remains to consider $\sum_{j = 1}^{j'-1} \delta_j$ and $\sum_{j = j'}^{m-1} \delta_j$. We shall bound them using Lemma~\ref{lem:cluster} and the proximity bound (Lemma~\ref{lem:proximity}).

    Consider $\Delta'_1, \Delta'_2, \ldots, \Delta'_{j'-1}$. We shall apply Lemma~\ref{lem:cluster} to them. Let $\Delta'_{\max} = \max_{j=1}^{m-1}\Delta'_j$ and let $\Delta'_{\min} = \min_{j=1}^{m-1}\Delta'_j$.  Note that $\Delta'_{\max} \leq 3/2$ since the items are from $[1,2] \times [1,2]$, and that \(
    \Delta'_{\min} \geq 1/(2\sum_j |I_j|) \geq  \frac{\alpha\eps}{2} \geq \frac{\eps}{4}
    \) due to Observation~\ref{obs:group-property}(ii), (iv), and (v). Therefore, $\log(\frac{\Delta'_{\max}}{\Delta'_{\min}}) = O(\log \frac{1}{\eps})$.  By Lemma~\ref{lem:cluster}, there exists $h = O(\log \frac{1}{\eps})$ indices $1 = j_1 < j_2 < \cdots < j_h = j^* - 1$ such that for any $k \in \{1, \ldots, h-1\}$, we have that
    \begin{equation}\label{eq:cluster}
            \max\{\Delta'_j : j_k < j < j_{k+1}\} \leq \sum\{\Delta'_j : j_{k+1} \leq j \leq j^*-1\}.
   \end{equation}
   Fix some $k$. Consider the groups $I_j$ with $j_k < j < j_{k+1}$. Let $\Delta' = \max\{\Delta'_j : j_k < j < j_{k+1}\}$. The inequality~\eqref{eq:cluster} implies that the items in these groups $I_j$ have efficiencies at least $\rho_b + \Delta'$. By Lemma~\ref{lem:proximity}, we have 
   \begin{equation}\label{eq:proximity-bound-b}
        \sum_{j_k < j < j_{k+1}} (w(I_j) - w^*_j) \leq \frac{2}{\Delta'}.  
   \end{equation}

   Also note that for each $I_j$ with $j_k < j < j_{k+1}$, we have $\frac{1}{|I_j|} \leq 2\Delta'_j \leq 2\Delta'$ (due to Observation~\ref{obs:group-property}(iv) and (v)). 
   Therefore, 
   \[
        \sum_{j_k < j < j_{k+1}} \delta_j \leq \sum_{j_k < j < j_{k+1}} \frac{w(I_j) - w^*_j}{\alpha|I_j|} \leq \sum_{j_k < j < j_{k+1}} \frac{2\Delta'(w(I_j) - w^*_j)}{\alpha} \leq \frac{1}{\alpha}.
   \]
   The last inequality is due to \eqref{eq:proximity-bound-b}.
   Recall that $\delta_j \leq \frac{1}{\alpha}$ for any $j$. We have the following.
   \[
        \sum_{j=1}^{j'-1} \delta_j  = \sum_{k=1}^h \delta_j  + \sum_{k=1}^{h-1}\sum_{j_k < j < j_{k+1}} \delta_j = O(\frac{h}{\alpha}) = \tilde{O}(\frac{1}{\alpha}).
   \]

   In a symmetric way, we can show that 
   \(
       \sum_{j=j'+ 1}^{m-1} \delta_j  \leq \tilde{O}(\frac{1}{\alpha}).
   \) 
\end{proof}

The above lemma immediately implies the following since $S_{\mathrm{head}}$ approximates $ S_1 \oplus \cdots \oplus S_m$ with factor $1 + O(\eps)$.
\begin{corollary}
    $S_{\mathrm{head}}$ approximates $(w^*, p^*)$ with additive error $\tilde{O}(\frac{1}{\alpha})$.
\end{corollary}

We summarize this subsection by the following lemma.
\begin{lemma}\label{lem:approx-head}
    In $\tilde{O}((\frac{1}{\eps})^{7/4})$ time and with high probability, we can compute a set of size $\tilde{O}(\frac{1}{\eps})$ that approximates $\mathcal{S}^+(I_{\mathrm{head}})$ with additive error $\tilde{O}(\frac{1}{\alpha})$.
\end{lemma}

\subsection{Approximating the Tail Group}
Recall that the tail group $I_{\mathrm{tail}}$ contains $\{(w_{\lceil \frac{1}{\alpha\eps}\rceil + 2}, p_{\lceil \frac{1}{\alpha\eps}\rceil + 2}), \ldots, (w_n, p_n)\}$. We partition $I_{\mathrm{tail}}$ into $m = O(\log \frac{1}{\eps})$ groups $I_0, \ldots, I_m$ so that $I_j$ is a maximal group of items whose efficiencies differ by at most $2^{j} \alpha\eps$. See Algorithm~\ref{alg:tail-group} for details.

\begin{algorithm}
\caption{$\mathtt{PartitionTailGroup}(I_{\mathrm{tail}})$}
\label{alg:tail-group}
    \begin{algorithmic}[1]
        \State $i = \lceil \frac{1}{\alpha\eps}\rceil + 2$, $j = 0$ and $I_j = \emptyset$
        \While{$i \leq n$}
            \If{the item efficiencies of $I_j\cup \{(w_i, p_i)\}$ differ by at most $2^{j-1}\alpha\eps$}
                \State $I_j := I_j \cup \{(w_i, p_i)\}$ and $i:= i+1$
            \Else
                \State $j:= j + 1$ and $I_j := \emptyset$ 
            \EndIf  
        \EndWhile
        \State m := j
        \State \Return $I_1, \ldots, I_{m}$
  \end{algorithmic}
\end{algorithm}

For each group $I_j$, we approximate it with accuracy parameter $2^j\eps$.  More precisely, we compute a set $S_j$ of size $\tilde{O}(\frac{1}{2^j\eps})$ that approximates each tuple $(w, p) \in \mathcal{S}^+(I_j)$ with additive error $2^j\eps w$ using Corollary~\ref{coro:approx-add} in time $\tilde{O}(|I_j| + \frac{1}{2^j\eps} + (\frac{1}{2^j\eps})^2 \cdot 2^j \alpha\eps)$.

Then we compute a set $S_{\mathrm{tail}}$ that approximates $S_0 \oplus \cdots \oplus S_m$ with factor $1 + O(\eps)$ via Lemma~\ref{lem:merging}. Note that the total size of $S_1, \ldots, S_m$ is $\tilde{O}(\frac{1}{\eps})$.  Since $m = O(\log \frac{1}{\eps})$, the time cost of Lemma~\ref{lem:merging} is $\tilde{O}((\frac{1}{\eps})^{3/2})$.

\paragraph*{Bounding Time Cost}
It is easy to see that the total time cost is 
\[
    \tilde{O}((\frac{1}{\eps})^{3/2}) + \sum_{j} \tilde{O}(|I_j| + \frac{1}{2^j\eps} + (\frac{1}{2^j\eps})^2 \cdot 2^j \alpha\eps) = \tilde{O}(n + (\frac{1}{\eps})^{3/2} + \frac{\alpha}{\eps}).
\]
Recall that $\frac{1}{2}\leq \alpha\leq (\frac{1}{\eps})^{3/4}$. The total time cost is $\tilde{O}(n + (\frac{1}{\eps})^{7/4})$.

\paragraph*{Bounding Additive Error}
Let $(w^*, p^*)$ be an arbitrary tuple in $\mathcal{S}^+(I; \frac{1}{\alpha\eps})$. Let $(w_b, p_b)$ be the breaking item with respect to $w^*$. Let $\rho_b = \frac{p_b}{w_b}$ be the efficiency of the breaking item. Note that $b \leq \frac{1}{\alpha\eps} + 1$, so $b \in I_{\mathrm{head}}$. Recall that $S_{\mathrm{head}}$ is the set we computed by Lemma~\ref{lem:approx-head}.

\begin{lemma}
    $S_{\mathrm{head}} \oplus S_0 \oplus S_1 \oplus \cdots \oplus S_m$ approximates $(w^*, p^*)$ with additive error $\tilde{O}(\frac{1}{\alpha})$.
\end{lemma}
\begin{proof}
    Since $I = I_{\mathrm{head}} \cup I_0 \cup \cdots \cup I_m$ and $(w^*, p^*) \in \mathcal{S}^+(I)$, we have that t
    \[
        (w^*, p^*) = (w^*_{\mathrm{head}}, p^*_{\mathrm{head}}) +  (w^*_0, p^*_0) + \cdots + (w^*_m, p^*_m).
    \]
    for some $(w^*_{\mathrm{head}}, p^*_{\mathrm{head}})  \in \mathcal{S}^+(I_{\mathrm{head}})$ and some $(w^*_j, p^*_j) \in \mathcal{S}^+(I_j)$ for each $j$. Let $\delta_j$ be the additive error with which $S_j$ approximates $(w^*_j, p^*_j)$. Since we already know that $S_{\mathrm{head}}$ approximates  $(w^*_{\mathrm{head}}, p^*_{\mathrm{head}})$ with additive error $\tilde{O}(\frac{1}{\alpha})$, to prove the lemma, it suffices to show that $\sum_j\delta_j \leq \tilde{O}(\frac{1}{\alpha})$.

    For $ j = 0$, we have that 
    \(
        \delta_0 = \eps w^*_0 \leq \eps w^* \leq \frac{1}{\alpha}. 
    \) 
    For $j \geq 1$, by the way we partition $I_{\mathrm{tail}}$,  the the items in $I_j$ have efficiencies at most $\rho_b - 2^{j-1}\alpha\eps$. By Lemma~\ref{lem:proximity}, $w^*_j \leq \frac{4}{2^j\alpha\eps}$. Therefore,
    \(
        \delta_j = 2^j \eps w^*_j \leq \frac{4}{\alpha}.
    \)
    Since $m = O(\log \frac{1}{\eps})$, we have $\sum_j\delta_j \leq \tilde{O}(\frac{1}{\alpha}) = \tilde{O}(\frac{1}{\alpha})$.
\end{proof}

The above lemma immediately implies the following since $S_{\mathrm{tail}}$ approximates $S_0 \oplus S_1 \oplus \cdots \oplus S_m$ with factor $1 + O(\eps)$.
\begin{corollary}
    $S_{\mathrm{head}} \oplus S_{\mathrm{tail}}$ approximates $(w^*, p^*)$ with additive error $\tilde{O}(\frac{1}{\alpha})$.
\end{corollary}

\subsection{Putting Things Together}
We summarize this section with the following lemma.
\begin{lemma}
   There is an $\tilde{O}({n+(\frac{1}{\eps})^{7/4}})$-time algorithm for the reduced problem, which is randomized and succeeds with high probability.
\end{lemma}

\end{document}